\documentclass[lettersize,journal]{IEEEtran}
\usepackage{amsmath,amsfonts}
\usepackage[ruled,linesnumbered]{algorithm2e}
\usepackage{bm} 
\usepackage{varwidth} 
\usepackage{array}
\usepackage{caption}
\usepackage{subcaption}
\usepackage{textcomp}
\usepackage{stfloats}
\usepackage{url}
\usepackage{verbatim}
\usepackage{graphicx}
\usepackage{amsthm}
\newtheorem{assumption}{Assumption}
\usepackage{amsmath, amssymb, amsthm}
\newtheorem{theorem}{Theorem}

\usepackage{orcidlink}
\usepackage{amssymb}
\usepackage{multirow}
\usepackage{xcolor}
\usepackage{float}
\usepackage{gensymb}
\usepackage{hyperref}
\usepackage{booktabs}
\usepackage{lettrine}
\usepackage[style=ieee, citestyle=numeric-comp, backend=biber]{biblatex}
\AtBeginBibliography{\small}
\begin{document}

\title{Actor-Identifier-Critic Reinforcement Learning for Adaptive Model-Free Optimal Control of Nonlinear Systems with Stochastic Packet Dropouts}


\author{Kianoush~Aqabakee\textsuperscript{\large\orcidlink{0009-0008-0450-5360}}, Kosar Behnia\textsuperscript{\large\orcidlink{0000-0002-1014-162X}}, Amirhossein~{Heydarian~Ardakani}\textsuperscript{\large\orcidlink{0000-0003-3989-5481}}, \\Farzaneh~Abdollahi\textsuperscript{\large\orcidlink{0000-0003-4957-987X}}, and Elham Shirazi\textsuperscript{\large\orcidlink{0000-0002-1274-1506}}
\thanks{K.~Aqabakee, K.~Behnia and F.~Abdollahi are with the Department of Electrical Engineering, Amirkabir University of Technology, Tehran, Iran, email: kianoush.aqabakee@aut.ac.ir, kosar.behnia@aut.ac.ir, f\_abdollahi@aut.ac.ir. A.H.~Ardakani and E.~Shirazi are with the Faculty of Engineering Technology, University of Twente, Enschede, Netherlands, email: a.heydarian@utwente.nl, e.shirazi@utwente.nl.}}

\markboth{Preprint,~Vol.~.., No.~.., March~2026}%
{Shell \MakeLowercase{\textit{et al.}}: A Sample Article Using IEEEtran.cls for IEEE Journals}

\maketitle

\begin{abstract}
Packet dropouts in control systems poses a critical challenge, as it can significantly compromise system performance and stability. In these conditions, classical controllers often struggle to deliver effective control, as they rely on accurate system models, which may not always be available.
This paper proposes a novel Actor-Identifier-Critic~(AIC) controller to address model-free tracking control of nonlinear systems in the presence of packet dropouts in both the controller-to-actuator and sensor-to-controller channels. 
Using an identifier to learn the system dynamics, the proposed controller is able to handle packet dropouts in the communication link and facilitate gradient propagation from the critic to the actor within a model-free control framework.
The performance of the proposed method is demonstrated on two nonlinear SIMO and MIMO systems and a case study on power system stability subject to stochastic packet dropouts.
\end{abstract}


\section{Introduction}\label{sec1}



\lettrine[lines=2]{O}{ptimal} control is a mathematical framework for guiding dynamic systems toward desired objectives, balancing performance and constraints.
These strategies have been applied to address a variety of challenges in control systems, such as model uncertainty~\cite{9416181, WEN2022368}, input/output saturation, and state constraints~\cite{9463406}.
However, these approaches rely on the availability of an accurate system model, which is often unavailable in real-world applications.


Researchers use approximation methods, such as neural networks~(NNs), to obtain optimal control strategies for uncertain dynamical systems.
In this regard, authors in~\cite{9397314} proposed a NN-based observer method for obtaining optimal control in the presence of external disturbances. To ensure observer stability, NN weights are updated based on the Lyapunov stability theorem.
Other NN-based optimal output feedback strategies are proposed in~\cite{9463406, Inputsat} for uncertain dynamical systems with input saturation and state constraints.
The authors in~\cite{AdaptiveFuzzy} estimated the uncertain dynamics of nonlinear systems using fuzzy logic systems and proposed an output feedback controller using the adaptive backstepping approach. 
A combined approach using backstepping and reinforcement learning~(RL) was introduced in~\cite{BackRL} to avoid the violation of state constraints for stochastic nonlinear systems.  
Authors in~\cite{critic} presented a strategy based on dynamic programming to find optimal control in an online manner. In this strategy, an NN was used to identify nonlinear system dynamics. Then, an online critic NN is designed to solve the Hamilton–Jacobi–Bellman~(HJB) equation. 


In controller design, it is often assumed that the transmitted data is accurate and that the communication links are reliable. However, in practice, different factors, such as packet dropouts, can compromise the reliability of communication networks~\cite{8128911}.
Packet dropouts in networked control systems occur due to network congestion, interference, or hardware malfunctions, leading to information loss that can degrade system performance and compromise stability~\cite{7920392}. Hence, effective strategies must be developed to mitigate data loss and ensure continuous, accurate control despite the unreliability of communication links.
To mitigate the impact of packet dropouts in linear control systems, researchers have used approaches, including model predictive control (MPC)~\cite{7501536,cai2019robust}, adaptive control~\cite{JIANG2024111690}, and fuzzy control~\cite{10366857}. Stochastic models, such as Bernoulli processes, have been used to further characterize communication uncertainties for robust control strategies~\cite{chen2023data,10332447}.
For nonlinear systems subject to packet dropouts, researchers have developed methods such as fuzzy approximation~\cite{10473199}, adaptive control~\cite{9802886, 10571932}, and sliding mode fault-tolerant control~\cite{10106049}.
However, these methods require at least partial knowledge of the system dynamics to be able to handle packet dropouts.



This paper addresses model-free tracking control of nonlinear systems with unknown dynamics in the presence of stochastic sensor and actuator packet dropouts by introducing an Actor-Identifier-Critic~(AIC) controller. 
The proposed controller incorporates an identifier to learn the system dynamics, allowing it to compensate for dropouts in the communication links between the plant and the controller. Moreover, the identifier facilitates gradient propagation from the critic to the actor within a model-free control framework.
The optimality of the controller is studied using the HJB equation, and its stability analysis is provided based on Lyapunov’s stability theorem.
The performance of the controller is evaluated on single-input multi-output~(SIMO) and multi-input multi-output~(MIMO) nonlinear systems and a power system case study under different stochastic dropout probabilities.
The main contributions of this paper are:
\begin{enumerate} 
    \item An AIC controller for model-free adaptive tracking control of nonlinear systems. The method employs an actor-identifier-critic structure, facilitating the simultaneous estimation and control of the system.
    
    \item Handling stochastic sensor and actuator packet dropouts through the estimation of missing data by using a stable identifier. It also plays an intermediary role between the critic and the actor for gradient propagation.
    
    \item Stability analysis of the controller based on Lyapunov's stability theorem.
\end{enumerate}


The remainder of the paper is structured as follows: Section~\ref{sec2} formulates the tracking control problem of nonlinear systems with dropouts. Section~\ref{sec3} presents the proposed RL-based optimal control strategy. Section~\ref{sec4} evaluates the performance of the proposed controller. Lastly, Section~\ref{sec5} concludes the paper and outlines future research directions.
\section{Problem Statement}\label{sec2}

\subsection{System Dynamic Model}
The continuous-time dynamic model of an affine nonlinear system is defined as follows~\cite{slotine1991applied}:

\begin{equation}
	\label{eq:sys_dynamic}
	\dot{x}(t) = f(x(t)) + g(x(t))u(t),
\end{equation}
where $t$ refers to time, $x \in \mathbb{R}^{n_x}$ represents the system states, $u \in \mathbb{R}^{n_u}$ denotes the plant input, and $f(\cdot) : \mathbb{R}^{n_x} \to \mathbb{R}^{n_x}$ and $g(\cdot) : \mathbb{R}^{n_x} \to \mathbb{R}^{n_x \times n_u}$ describe the unknown system dynamics. 

By  adding
and subtracting $A_cx(t)$ from~\eqref{eq:sys_dynamic}, the following can be obtained:
\begin{align}
  \label{eq:msystem}
        &\dot{x}(t) = A_cx(t) + f_c(x(t)) + g(x(t))u(t),\\[5pt]
        &f_c(x(t)) = f(x(t)) - A_cx(t) , 
\end{align}
where $A_c$ is a diagonal Hurwitz matrix.

To apply the state transition property of a Markov process, we use the derivative definition based on the Euler’s methods~\cite{biswas2013discussion}.
The time derivative of the system states are represented as: 
\begin{equation}
\label{eq:sample_time}
    \dot{x}(t) = \lim_{\Delta t \to 0} \frac{x(t + \Delta t) - x(t)}{\Delta t}.
\end{equation}
Using this definition, future system states can be predicted based on the current states and their derivatives. 
This is crucial for calculating the expected future rewards using the Bellman equation. 
The evolution of states over time is represented as:
\begin{equation}
    \label{eq:x_future}
    x(t+\Delta t) = x(t) + \int_{t}^{t+\Delta t} \dot{x}(\tau) d\tau.
\end{equation}

The following assumptions are considered throughout this article:


\begin{assumption}\label{as:0}
    The nonlinear system in~\eqref{eq:sys_dynamic} is observable.
\end{assumption}

\begin{assumption}\label{as:1}
    There exists an unknown constant term $\bar{g}$ such that $\parallel g(\cdot)\parallel  \le \bar{g}$.
\end{assumption}

\begin{assumption}\label{as:2}
	The system dynamics $f(\cdot)$ and $g(\cdot)$ in~\eqref{eq:sys_dynamic} are assumed to be continuous Lipschitz functions.
\end{assumption}

Assumption~\ref{as:0} is not conservative and many nonlinear systems are observable~\cite{1677582, 7460133}. Assumptions~\ref{as:1} and~\ref{as:2} are general and non-restrictive assumptions. Numerous nonlinear systems satisfy Lipschitz continuous condition~\cite{10387713, 9802886} and in many practical systems, $g(\cdot)$ remains bounded~\cite{esfandiari2022neural, 10387713}.


Considering the above mentioned assumptions, evolution of the system state $x(t)$ over small, finite time steps $\Delta t > 0$ is approximated as:
\begin{equation}
    x(t + \Delta t) \approx x(t) + \dot{x}(t) \Delta t,
\end{equation}
where the smoothness of $\dot{x}$(t) is concluded by Assumption~\ref{as:2}. This ensures that the state evolution remains finite and well-defined.

A packet dropout is modeled by a stochastic variable~$\gamma$, which determines whether the signal is received on the plant side $\gamma_s$ or on the controller side $\gamma_c$.
When a dropout occurs on the plant side, the plant state received by the controller can be described as:
\begin{align}
x_{s}(t) = \gamma_s x(t),
\end{align}
where $x_{s} \in \mathbb{R}^{n_x}$ represents the plant state feedback received on the controller side, and $\gamma_s$ is a Bernoulli variable defined as follows:
\begin{equation}
	\gamma_s = \begin{cases}
	1 & \text{No dropout} \\
	0 & \text{Dropout}.
	\end{cases}
\end{equation}
To model a dropout on the controller side, the system input (i.e., the control command on the plant side) is expressed as:
\begin{equation}
u(t) = \gamma_c u_c(t),
\end{equation}
where $u_c \in \mathbb{R}^{n_u}$ is the controller output and  $\gamma_c \in \mathbb{R}$  is a Bernoulli variable as follows:
\begin{equation}
	\gamma_c = \begin{cases}
	1 & \text{No dropout} \\
	0 & \text{Dropout}.
	\end{cases}
\end{equation}
The dropout events are modeled by the probabilities $P(\gamma_s=1)=\bar{\gamma}_s$ and $P(\gamma_c=1)=\bar{\gamma}_c$, indicating the stochasticity of dropouts from the plant sensor and the controller, respectively.

\subsection{Tracking Control Formulation}\label{subsection:Optimization_Formulation}

The tracking control objective is to generate an action~$\hat{u}_c(t)$ that ensures the system state~$x(t)$ defined in~\eqref{eq:sys_dynamic}, tracks a given reference trajectory $x_d(t)$. In this regard, a quadratic cost function is considered based on the tracking error and the control command.
The cost function for the tracking problem is given by:

\begin{equation}
\label{eq:J}
    J(e(t), \hat{u}_c(t))  = \int_{t}^{\infty}  \left[ e(\tau)^T Q e(\tau) + \hat{u}_c(\tau)^T R \hat{u}_c(\tau)\right] d\tau,
\end{equation}
where $e(t) = x(t)-x_d(t)$ is the tracking error and $Q$ and $R$ are positive definite matrices to weight the tracking error and the control command, respectively.

\begin{assumption}\label{as:3}
	The desired trajectory $x_d$ and its time derivative are smooth and bounded. 
\end{assumption}

The value function for the tracking problem can be defined based on the Bellman equation~\cite{9802886}:
\begin{equation}
\label{eq:value_function}
    \begin{aligned}    
    \mathbf{V}(e(t)) = \mathbb{E} \big\{ \int_{t}^{t+\Delta t}  \left[ e(\tau)^T Q e(\tau) + \hat{u}_c(\tau)^T R \hat{u}_c(\tau)\right] d\tau\\
    + \mathbf{V}(e(t+\Delta t)) \mid e(t) \big\}.
    \end{aligned}
\end{equation}
Using stochastic control theory~\cite{9802886}, this value function is expanded to consider packet dropouts:
\begin{equation}
    \label{eq:optimization_with_dr}
    \begin{aligned}
        \mathbf{V}(e(t), \hat{u}_c(t)) =& \ \int_t^{t+\Delta t} \big[ e(\tau)^T Q e(\tau) + \hat{u}_c(\tau)^T R \hat{u}_c(\tau) \big] d\tau \\
        &+ \bar{\gamma}_c \mathbf{V}(e(t+\Delta t)) \Bigg|_{u(t) = u_c(t)} \\
        &+ (1 - \bar{\gamma}_c) \mathbf{V}(e(t+\Delta t)) \Bigg|_{u(t) = \mathbf{0}},
    \end{aligned}
\end{equation}
where the next-step tracking error can be obtained based on the system model in~\eqref{eq:sys_dynamic}:
\begin{equation}
    \label{eq:next_step_error}
    \begin{aligned}
        e(t+\Delta t) &= x(t+\Delta t)-x_d(t+\Delta t) \\[5pt] & = x(t) + \int_{t}^{t+\Delta t} \dot{x}(\tau) d\tau - x_d(t+\Delta t) .
    \end{aligned}
\end{equation}

\section{Actor-Identifier-Critic Controller}\label{sec3}
Packet dropouts in communication links pose a significant challenge to networked control systems, undermining their controllability, especially when the system dynamics are unknown.
To control uncertain nonlinear systems in the presence of packet dropouts, an Actor-Identifier-Critic (AIC) controller is proposed based on a model-free adaptive RL. 
As illustrated in Fig.~\ref{fig:main_diag}, the proposed AIC controller consists of three interconnected components:
\begin{itemize}
    \item \textbf{Actor} generates control actions based on the system tracking error.
    \item \textbf{Identifier} estimates the system dynamics and creates a mathematical link between other components.
    \item \textbf{Critic} evaluates the performance of the control actions and guides the learning process by providing feedback.
\end{itemize}
This triple-network control framework ensures reliable performance even in the presence of sensor and actuator dropouts. The following sections provide a detailed explanation of the AIC controller structure and learning principles for each component.
\begin{figure}[t]
    \centering
    \includegraphics[width=1\linewidth]{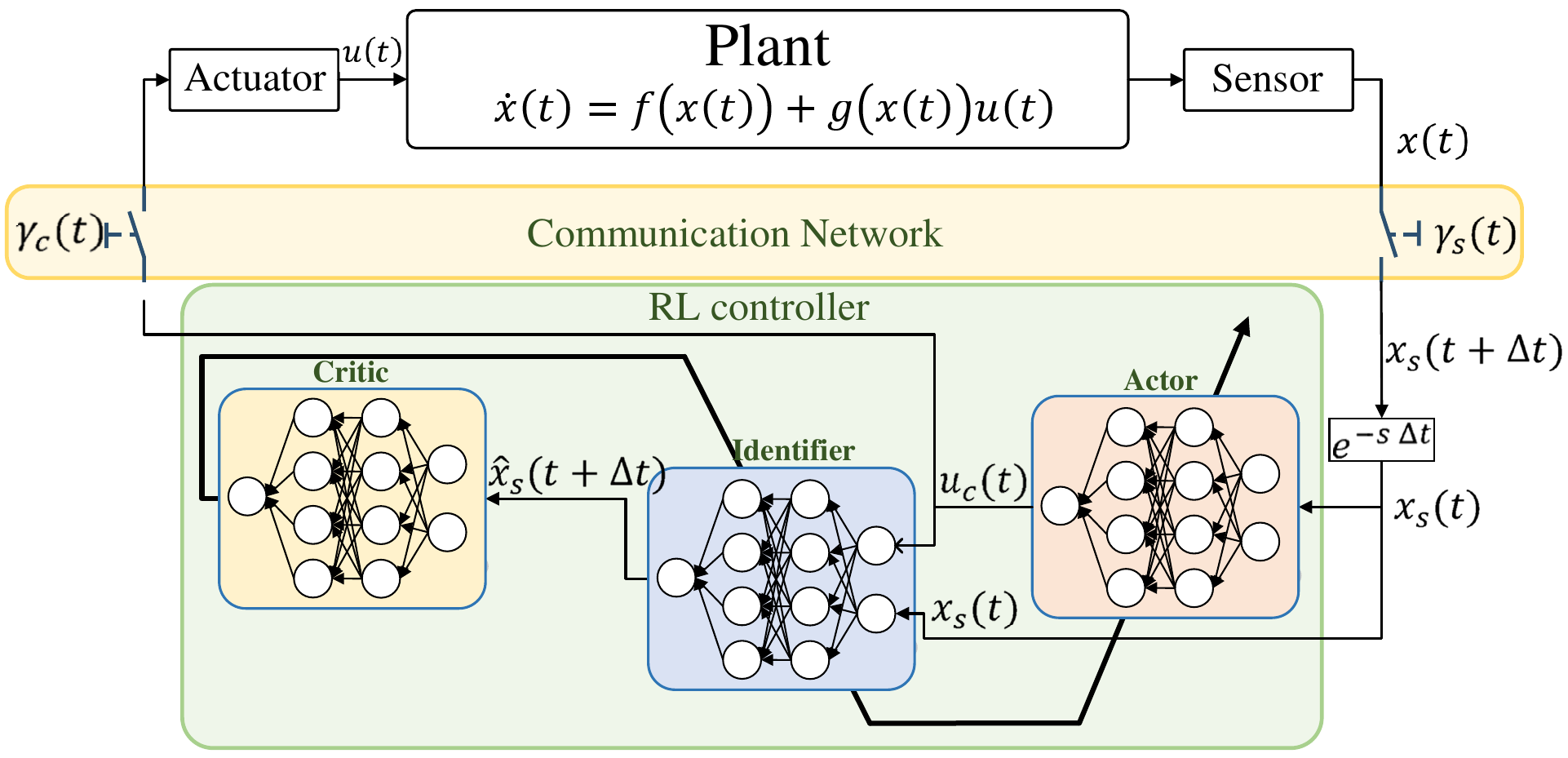}
    \caption{Schematic of the proposed AIC controller.}
    \label{fig:main_diag}
\end{figure}

\subsection{AIC Controller Components}
As shown in Fig.~\ref{fig:main_diag}, the proposed AIC controller is connected to the plant through a communication network that is subject to packet dropout. The controller consists of three interconnected components, namely actor, identifier, and critic networks. The actor receives an estimated state $x_s(t)$ and generates the control input~$u_c(t)$. The identifier predicts the next state $\hat{x}_s(t+\Delta t)$ to allow the critic estimate the value function and accordingly update the actor.

As shown in Fig.~\ref{fig:main_diag}, the proposed AIC controller is connected to the plant through a communication network subject to packet dropouts. The controller comprises three interconnected components: the actor, the identifier, and the critic networks. The actor receives the estimated state~$x_s(t)$ and generates the control command~$u_c(t)$ and the identifier predicts the next state of the system~$\hat{x}_s(t+\Delta t)$, allowing the critic to estimate the value function and thenccordingly update the actor.
The following describes the formulation of these components.

\subsubsection{Identifier}
The identifier is responsible for estimating the nonlinear system dynamics in~\eqref{eq:msystem}.
The identifier is defined as a function approximator with bounded estimation error as follows:
\begin{equation}
    \label{eq:identifier_part}
            f_c(x(t)) + g(x(t))u_c(t) = \mathcal{F}_i(x(t),  u_c(t)) + \epsilon_i,
\end{equation}
where $\mathcal{F}_i$ is the function approximator for system dynamics and $\epsilon_i \in \mathbb{R}^{n_x}$ represents the bounded estimation error.
The identifier serves as a link between actor and critic to transfer gradients. Therefore, it is crucial for the identifier to be a stable model. In this study, we use an identifier, proposed in~\cite{talebi2009neural}, which is based on a multi-layer perceptron~(MLP) with stable update rules for its parameters. 
The used identifier is described as follows:
\begin{equation}
\label{identifier_structure}
     \mathcal{F}_i(x_s(t),u_c(t)) = \hat{W}_i \sigma(\hat{V}_i\bar{x}(t)),\vspace{5px}
\end{equation}
where $\hat{W}_i$, $\hat{V}_i$ represent the estimated weights for the output and hidden layers,~$\bar{x}(t) = \left[x_s(t),u_c(t)\right]$ is the input of the identifier and,~$\sigma(\cdot)$ is a bipolar sigmoid function defined as:
\begin{equation}
    \small
    \sigma(x) = \frac{2}{1 + e^{-x}} - 1.
\end{equation}
It is worth mentioning that, according to the universal approximation theorem, any continuous nonlinear function defined on a compact set can be approximated using an NN with a single hidden layer and a sufficient number of neurons~\cite{HORNIK1989359, esfandiari2022neural}.


To solve~\eqref{eq:optimization_with_dr} and determine the optimal control policy, the next-step tracking error $e(t+\Delta t)$ is required in the state value function. However, considering that the system dynamics is unknown in~\eqref{eq:next_step_error}, the next-step system state is first estimated using the identifier network to obtain the next-step tracking error.
Substituting~\eqref{identifier_structure} in system dynamics~\eqref{eq:msystem}, the estimated state dynamics can be described as follows:
\begin{equation}
	\begin{aligned}
	\dot{x}_s(t) & \approx A_c x_s(t) + \mathcal{F}_i(x_s(t),u_c(t)).
	\end{aligned}
\end{equation}
Then, since $\hat{e}_s=x_s(t)-x_d(t)$ the estimated tracking error dynamics can be obtained as:
\begin{equation}
\label{eq:identifier_main_rule}
\begin{aligned}
    \dot{\hat{e}}_s(t) &= \dot{x}_s(t)-\dot{x}_d(t)\\
    &=A_c x_s(t) + \mathcal{F}_i(x_s(t),u_c(t)) - \dot{x}_d(t).
\end{aligned}
\end{equation}
Therefore, the estimated next-step tracking error is obtained as:
\begin{equation}
    \small
    \label{eq:identifier_output}
    \begin{aligned}
    \hat{e}_s(t + \Delta t) = x_s(t) + \int_{t}^{t+\Delta t} (A_c x_s(\tau) + \mathcal{F}_i(x_s(\tau),u_c(\tau)))d\tau\\
    - x_d(t + \Delta t)
    \end{aligned}
\end{equation}
The estimated next-step tracking error~$\hat{e}_s(t + \Delta t)$ is used both to handle sensor-to-controller dropouts in the next control step and to allow gradient propagation for updating the actor parameters (see Section~\ref{sec:actor_update}).

\subsubsection{Critic}
The critic network estimates the unknown value function in~\eqref{eq:optimization_with_dr}, based on the tracking error resulting from the control command. This value space is estimated using a weighted summation of predefined basis functions, which is defined as:
\begin{equation}
    \label{eq:critic_forward}
    \hat{\mathbf{V}}(t) = \mathcal{F}_c(\hat{e}_s(t)) = \hat{W}_c\Psi_c(\hat{e}_s(t))
\end{equation}
Here, $\hat{W}_c$ represents the critic weight matrix, and $\Psi_c$ is the vector of basis functions, which can be defined as:
\begin{equation}
\Psi_c(x)=
\left [ 
\begin{matrix}
\psi_{c_1}(\cdot) & \psi_{c_2}(\cdot) & \cdots &
\end{matrix}
\right ]^T
\end{equation}
The critic network is used to improve the control policy and obtain the update rules of the actor network.

\subsubsection{Actor}
The actor network (i.e., the control policy) generates the control command based on the tracking error. 
The control command is the output of the actor network, which is defined as:
\begin{equation}
    \label{actor}
    u_c^*(t) = \mathcal{F}_a(\hat{e}_s(t)) + \epsilon_a,
\end{equation}
where $u_c^* \in \mathbb{R}^{n_u}$ represents the control output by the actor. The function $\mathcal{F}_a$ denotes a function approximator as an MLP, and~$\epsilon_a \in \mathbb{R}^{n_u}$ is the bounded estimation error of the actor.
Therefore, the actor network is described as follows:
\begin{align}
\label{actor_mat}
    u_c^*(t) = W_a \sigma\left(V_a \hat{e}_s(t)\right) +\epsilon_a,
\end{align}
where $W_a$ and $V_a$ denote ideal weights for output and hidden layers, respectively. Since the ideal weights are unknown, an approximation of these weights is used in the actor network. Then, the structure of the actor network can be expressed as follows:
\begin{align}
\label{actor_matrix}
    \hat{u}_c(t) = \hat{W}_a \sigma\left(\hat{V}_a \hat{e}_s(t)\right)
\end{align}
here, $\hat{u}_c(k) \in \mathbb{R}^{n_u}$, $\hat{W}_a$ and $\hat{V}_a$ represent the estimated control action and weights for output and hidden layers, respectively.



\subsection{Learning Principles}
In the following, we present the learning rules for updating the AIC controller components, i.e., actor, identifier, and critic networks. These update rules ensure the convergence of the networks in the presence of stochastic sensor and actuator dropouts.
\subsubsection{Identifier Update}
The identifier network estimates the system dynamics to compensate for packet dropouts in sensor data and create a gradient link between the actor and critic updates. 
The identifier function $\mathcal{F}_i(\cdot)$ in~\eqref{identifier_structure} uses the following update rules with guaranteed stability proof as provided in~\cite{talebi2009neural}:
\begin{equation}
	\label{eq:identifier_update_1}
	\begin{aligned}
		\dot{\hat{W}}_i = & \eta_{i_1} \left(\tilde{x}(t)^T A_c^{-1} \right)^T \left(\sigma(V_i\bar{x}(t))\right)^T - \rho \|\tilde{x}(t)\|\|\hat{W}_i\|,
	\end{aligned}
\end{equation}
\begin{equation}
	\label{eq:identifier_update_2}
	\begin{aligned}
		\dot{\hat{V}}_i = & \eta_{i_2} \left(\tilde{x}(t)^T A_c^{-1} \right)^T \left(\hat{W}_i (I - \Pi_i(t))\right)^T -\rho \|\tilde{x}(t)\|\|\hat{V}_i\|. 
	\end{aligned}
\end{equation}
here, $\tilde{x}(t) = x_s(t) - \hat{x}_s(t)$ is the identifier state estimation error, $\sigma(\cdot)$ represents the sigmoid activation function, $\eta_{i_1}$ and $\eta_{i_2}$ denote the learning rates, $\rho$ is the regularization parameter, $I$ is the identity matrix with appropriate dimension, and $\Pi_i(t) = \text{diag}(\sigma^2(V^T_i \bar{x}(t)))$ is calculated based on sigmoid derivative.

Due to the presence of stochastic dropouts on the controller side, the identifier network estimates future states based on a weighted summation of both dropout and non-dropout scenarios, expressed as:
\begin{equation}
    \small
    \label{eq:identifier_update_error}
    \begin{aligned}
    	\hat{x}_s(t+\Delta t) = & \bar{\gamma}_c \hat{x}_s(t+\Delta t) \Bigg|_{u(t) = \hat{u}_c(t)} + (1 - \bar{\gamma}_c) \hat{x}_s(t+\Delta t) \Bigg|_{u(t) = \mathbf{0}}.
    \end{aligned}
\end{equation}
The estimated state is then backpropagated in~\eqref{eq:identifier_update_1} and~\eqref{eq:identifier_update_2} to update the identifier network parameters. The stochastic forward calculation in the training process is illustrated in Fig.~\ref{fig:identifier_update}.
\begin{figure}[H]
    \centering
    \includegraphics[width=1\linewidth]{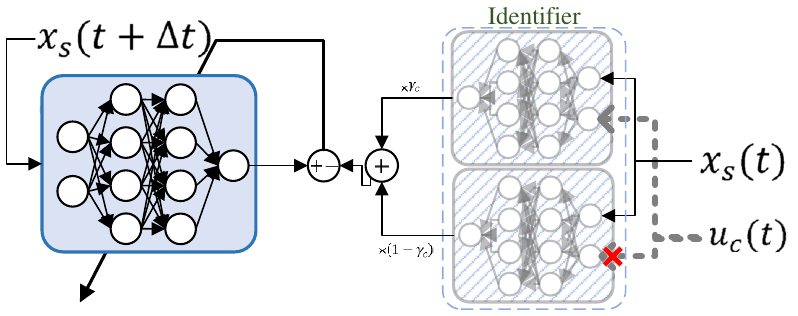}
    \caption{Stochastic forward calculation for updating identifier parameters.}
    \label{fig:identifier_update}
\end{figure}

\subsubsection{Critic Update}
To solve the nonlinear control problem described in Section~\ref{subsection:Optimization_Formulation}, we define a model-free Bernoulli Hamilton–Jacobi–Bellman~(BHJB) equation that considers packet dropouts. The BHJB is defined as follows:
\begin{equation}
\label{eq:final_HJ_track}
    \begin{aligned}
        \hat{\mathbf{V}}(t) &= \mathcal{F}_c(\hat{e}_s(t)) = \\
        &\int_t^{t+\Delta t} \big[ \hat{e}_s(\tau)^T Q \hat{e}_s(\tau) + \hat{u}_c(\tau)^T R \hat{u}_c(\tau) \big] d\tau \\ &+ \bar{\gamma}_c \mathcal{F}_c (\hat{e}_s(t+\Delta t)) \Bigg|_{u(t) = \hat{u}_c(t)} \\ 
        &+ (1-\bar{\gamma}_c) \mathcal{F}_c (\hat{e}_s(t+\Delta t)) \Bigg|_{u(t) = \mathbf{0}}. 
    \end{aligned}
\end{equation}
When calculating future value estimates, the BHJB considers the stochastic nature of the system by using a weighted summation of both scenarios with and without dropout.

Based on the proposed BHJB, the update rule for the critic network weights can be calculated as follows:
\begin{equation}
\label{eq:critic_weight_update_rule}
    \dot{\hat{W}}_c = \eta_c \hat{\mathcal{E}}_{TD} \cdot \Psi_c(\hat{e}_s(t))^T,
\end{equation}
where $\dot{\hat{W}}_c$ represents the time derivative of the critic network weights. The term $\hat{\mathcal{E}}_{TD}$ is the temporal difference error, which is obtained as:
\begin{equation}
\label{eq:td_error}
\begin{aligned}
    \hat{\mathcal{E}}_{TD} = \int_t^{t+\Delta t} \big[ \hat{e}_s(\tau)^T Q \hat{e}_s(\tau) + \hat{u}_c(\tau)^T R &\hat{u}_c(\tau) \big] d\tau \\ + \bar{\gamma}_c \mathcal{F}_c (\hat{e}_s(t+\Delta t)) &\Bigg|_{u(t) = \hat{u}_c(t)} \\ 
    + (1-\bar{\gamma}_c) \mathcal{F}_c (\hat{e}_s(t+\Delta t)) &\Bigg|_{u(t) = \mathbf{0}} \\ 
        - \mathcal{F}_c(\hat{e}_s(t))&.
\end{aligned}
\end{equation}

The critic network update process based on the BHJB is illustrated in Fig.~\ref{fig:critic_update}.
\begin{figure}[H]
    \centering
    \includegraphics[width=1\linewidth]{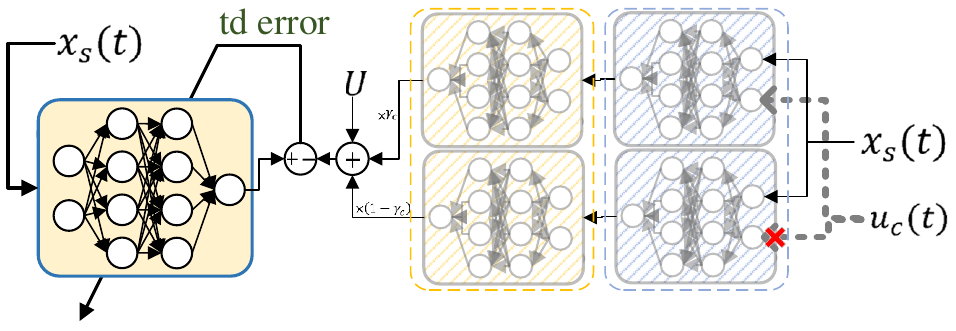}
    \caption{Stochastic forward calculation for updating critic parameters.}
    \label{fig:critic_update}
\end{figure}

\subsubsection{Actor Update}
\label{sec:actor_update}
The actor network generates control command $\hat{u}_c(t)$ based on the tracking error $\hat{e}_s(t)$. The actor objective is to minimize the estimated value $\hat{\mathbf{V}}$ and ensure that the system follows the desired trajectory. To ensure optimality of the control policy, let us first define the Hamiltonian function as follows:
\begin{equation}
	\label{eq:hamiltonian}
	\begin{aligned}
		\mathcal{H}(&e_s(t), \hat{u}_c(\cdot), \mathbf{V}(\cdot)) = \\
		& \hat{e}_s(t)^T Q \hat{e}_s(t) + \hat{u}_c(t)^T R \hat{u}_c(t) + \nabla \mathbf{V}(\hat{e}_s(t))^T \dot{\hat{e}}_s(t).
	\end{aligned}
\end{equation}
Here, $\nabla \mathbf{V}(\cdot)$ is the derivative of the value space estimate with respect to $\hat{e}_s(t)$.
The optimal value function and control policy should satisfy the HJB equation. The Hamiltonian function for the estimated optimal control policy can be described as follows:
\begin{equation}
	\label{eq:hamiltonian_optimal}
	\begin{aligned}
		\mathcal{H}(&e_s(t), u^*_c(\cdot), \mathbf{V}^*(\cdot)) = \\
		& \hat{e}_s(t)^T Q \hat{e}_s(t) + u^*_c(t)^T R u^*_c(t) \\&+ \nabla \mathbf{V}^*(\hat{e}_s(t))^T \big( A_c x_s(t) + \mathcal{F}_i(x_s(t),u^*_c) - \dot{x}_d(t) \big) = 0
	\end{aligned}
\end{equation}
Based on Bellman optimality, the optimal action can be obtained as follows~\cite{10015036}:
\begin{equation}
    \label{eq:open_hjb}
	\begin{aligned}
		u^*_c(t) &= \arg\min_{\theta_a} \big( \mathcal{H}(\hat{e}_s(t), \hat{u}_c(\cdot), \mathbf{V}^*(\cdot)) \big) \\& = \arg\min_{\theta_a} \bigg(
		\int_{t}^{t+\Delta t} (\hat{e}_s(t)^T Q \hat{e}_s(t) + \hat{u}_c(\tau)^T R \hat{u}_c(\tau))d\tau \\ & \quad \quad \quad \quad\quad \quad\quad \quad+ \mathbf{V}^*(\hat{e}_s(t+\Delta t)) 
		\bigg),
	\end{aligned}
\end{equation}
where $\theta_a \in \{ W_a, V_a \}$ is the set of actor weight matrices.
To minimize this objective, the parameters of the actor network are updated through a chain gradient process involving the critic and the identifier (see Fig.~\ref{fig:main_diag}). By calculating the derivative of~\eqref{eq:open_hjb} with respect to actor parameters, the actor update rule can be derived as:
\begin{equation}
	\dot{\theta}_a = \eta_{a} \bigg[ 2 R \hat{u}_c(t)\int_{t}^{t+\Delta t} 1d\tau 
	+ \frac{\partial \hat{\mathbf{V}}(t+\Delta t)}{\partial \hat{u}_c(t)}
	 \bigg] \frac{\partial \hat{u}_c(t)}{\partial \theta_a}.
\end{equation}
The term $\frac{\partial \hat{\mathbf{V}}(t+\Delta t)}{\partial \hat{u}_c(t)}$ can be calculated as:
\begin{equation}
	\begin{aligned}
		&\frac{\partial \hat{\mathbf{V}}(t+\Delta t)}{\partial \hat{u}_c(t)} = \\[7pt]
		&\frac{\partial \hat{\mathbf{V}}(t+\Delta t)}{\partial \hat{e}_s(t+\Delta t)} \cdot 
		\frac{\partial \hat{e}_s(t+\Delta t)}{\partial \mathcal{F}_i(x_s(t),\hat{u}_c(t))} \cdot 
		\frac{\partial \mathcal{F}_i(x_s(t), \hat{u}_c(t))}{\partial \hat{u}_c(t)}\\[7pt]
	\end{aligned}
\end{equation}
First term $\frac{\partial \hat{\mathbf{V}}(t+\Delta t)}{\partial \hat{e}_s(t+\Delta t)}$ describes the critic network gradients which is obtained as:
\begin{equation}
	\frac{\partial \hat{\mathbf{V}}(t+\Delta t)}{\partial \hat{e}_s(t+\Delta t)} = W_c \Psi_c' \left( \hat{e}_s(t+\Delta t)\right).
\end{equation}
Moreover, based on~\eqref{eq:identifier_output}, it can be concluded that:
\begin{equation}
		\frac{\partial \hat{e}_s(t+\Delta t)}{\partial \mathcal{F}_i(x_s(t),\hat{u}_c(t))} = \int_{t}^{t+\Delta t} 1d\tau = \Delta t.
\end{equation}
Using~\eqref{identifier_structure}, the gradients of the identifier network with respect to the input~$\hat{u}_c$ are represented as:
\begin{equation}
	\frac{\partial \mathcal{F}_i \left(x_s(t), \hat{u}_c(t)\right)}{\partial \hat{u}_c(t)} = \hat{W}_i \left( I - \Pi_i(t) \right)  \hat{V}_i.
\end{equation}
Lastly, based on~\eqref{actor_matrix}, the gradients of the actor network with respect to its parameters are given by:
\begin{equation}
\label{actor_update}
	\left\{
	\begin{aligned}
		\frac{\partial \hat{u}_c(t)}{\partial \hat{W}_a} &= \sigma \left( \hat{V}_a \hat{e}_s(t) \right) \\
		\frac{\partial \hat{u}_c(t)}{\partial \hat{V}_a} &= \hat{W}_a \left( I - \Pi_a \right)
	\end{aligned}
	\right.
\end{equation}
where $\left( I - \Pi_a \right)=I -\text{diag} \left( \sigma^2 \left( \hat{V}_a \hat{e}_s(t) \right) \right)$ is the derivative of the bipolar Sigmoid function with respect to its input. The final update rules for the actor network parameters are obtained as follows:
\begin{equation}
\label{eq:actor_update_1}
\small
	\begin{aligned}
		\dot{\hat{W}}_a = \eta_{a_1} \Delta t \Bigg[& \hat{W}_c \cdot \Psi_c' \left( \hat{e}_s(t + \Delta t) \right) \cdot \hat{W}_i \cdot \left( I - \Pi_i(t) \right) \cdot \hat{V}_i \\
		& + R \hat{u}_c(t) \Bigg] \cdot \sigma \left( \hat{V}_a \hat{e}_s(t) \right) 
	\end{aligned}
\end{equation}
\begin{equation}
\label{eq:actor_update_2}
\small
	\begin{aligned}
		\dot{\hat{V}}_a = \eta_{a_2} \Delta t \Bigg[&  \hat{W}_c \cdot \Psi_c' \left( \hat{e}_s(t + \Delta t) \right) \cdot \hat{W}_i \cdot \left( I - \Pi_i(t) \right) \cdot V_i \\
		& + R \hat{u}_c(t) \Bigg] \cdot \hat{W}_a \cdot \left( I - \Pi_a(t) \right) 
	\end{aligned}
\end{equation}

The pseudocode for the proposed AIC controller is outlined in Algorithm~\ref{alg:method}.

\begin{algorithm}[h]
    \caption{AIC controller for online optimal model-free tracking control}
    \label{alg:method}
    \SetKwBlock{Input}{Initialization}{end}

    \begin{varwidth}{\linewidth}
    \textbf{Input:} Initialize network parameters:
    \begin{itemize}
        \item Critic: $\hat{W}_c$
        \item Actor: $\hat{W}_a$, $\hat{V}_a$
        \item Identifier: $\hat{W}_i$, $\hat{V}_i$
    \end{itemize}
    \end{varwidth}
    \\[5pt]
        
    \For{each $x_d(t)$}{
       Output the control command~\eqref{actor_matrix}:
       \begin{equation*}
           \resizebox{0.4\hsize}{!}{$
           \hat{u}_c(t) = \hat{W}_a \sigma\left(\hat{V}_a \hat{e}_s(t)\right)$}
       \end{equation*}
        
        Estimate the future tracking error using the identifier~\eqref{eq:identifier_output}:
        \begin{equation*}
        \resizebox{1\hsize}{!}{
        $\begin{aligned}
            \hat{e}_s(t + \Delta t) = x_s(t) + \int_{t}^{t+\Delta t} (A_c x_s(\tau) + \mathcal{F}_i(x_s(\tau),u_c(\tau)))d\tau\\
            - x_d(t + \Delta t)
        \end{aligned}$}
        \end{equation*}
        
        \textbf{Update the critic} weight matrices $\hat{W}_c$ using the estimated error based on~\eqref{eq:critic_weight_update_rule} and~\eqref{eq:td_error}.\\[5pt]
        
        Estimate the system state with stochastic dropout~$\gamma_c$ based on~\eqref{eq:identifier_update_error}:
        \begin{equation*}
            \resizebox{1\hsize}{!}{$
            \hat{x}_s(t+\Delta t) = \bar{\gamma}_c \hat{x}_s(t+\Delta t) |_{u(t) = \hat{u}_c(t)} + (1 - \bar{\gamma}_c) \hat{x}_s(t+\Delta t) |_{u(t) = \mathbf{0}}$}
        \end{equation*}

        \textbf{Update the identifier} weight matrices $\hat{W}_i$ and $\hat{V}_i$ using the estimated state based on~\eqref{eq:identifier_update_1}, \eqref{eq:identifier_update_2}.\\[5pt]
        
        \textbf{Update the actor} weight matrices $\hat{W}_a$ and $\hat{V}_a$ based on~\eqref{eq:actor_update_1}, \eqref{eq:actor_update_2}.
        \\[5pt]
    }
    
\end{algorithm}

\subsection{Stability Analysis}

To analyze the stability of the proposed AIC controller, the following Lyapunov function is considered:
\begin{equation}
	L = L_a + L_i + L_c,
\end{equation}
where $L_a$, $L_i$, and $L_c$ are Lyapunov functions for the actor, identifier, and critic respectively.

First, in Theorem~\ref{theorem:id1}, the stability of the identifier network is analyzed without considering the packet dropouts.
\begin{theorem}
	\label{theorem:id1}
	For system~\eqref{eq:msystem}, if the weight matrices of the identifier network is updated by~\eqref{eq:identifier_update_1} and \eqref{eq:identifier_update_2} with $0 < \eta_i < 2$, the weight estimation error $\tilde{W}_i= W_i-\hat{W}_i$ is UUB.
\end{theorem}
The proof for Theorem~\ref{theorem:id1} is provided in~\cite{talebi2009neural}.
In Theorem~\ref{theorem:identifier}, the stability of the system is analyzed in the presence of stochastic dropouts.
\begin{theorem}
	\label{theorem:identifier}
	 The identifier described in (\ref{identifier_structure}) estimates the system output defined in (\ref{eq:sys_dynamic}) with a bounded error in the presence of the packet dropouts, if the closed-loop system is stable.
\end{theorem}
\begin{proof}
	For $\gamma_c=1$ and $\gamma_c=0$, in~\cite{talebi2009neural} it has been proved that the identifier estimates the state of the system with bounded error. For $0< \gamma_c < 1$, the estimation error can be described as a convex combination of cases~$\gamma_c=1$ and $\gamma_c=0$. Since boundedness is preserved under convex combinations of stable systems, the identifier estimates the system states with bounded error for all $0< \gamma_c < 1$.    
\end{proof}

To guarantee the stability of the critic network, we consider the following assumption:
\begin{assumption}
    \label{as:critic_terms}
    In~\eqref{eq:critic_weight_update_rule}, we assume that $\left\| \Psi_c(\hat{e}_s(t)) \right\| \leq \Psi_M$ and $\left\| W_c\right\| \leq W_M$, where $\Psi_M$ and are $W_M$ positive constants~\cite{wang2019adaptive}.
\end{assumption}
\begin{theorem}
	For the critic network in~\eqref{eq:critic_forward}, if the weight vector~$\hat{W}_c$ is updated using~\eqref{eq:td_error} with $0 < \eta_c < 2$, the estimation error $\tilde{W}_c=W_c-\hat{W}_c$ is UUB.
\end{theorem}
\begin{proof}
	 This proof is based on the Lyapunov stability theorem. In this regard, the Lyapunov candidate is defined as:
	\begin{equation}
		L_c(t) = \frac{1}{2 \eta_c} \text{tr} \{ \tilde{W}_c^T \tilde{W}_c \}.
	\end{equation}
	where $\tilde{W}_c = W_c - \hat{W}_c$ is the estimation error for critic weights. The derivative of the Lyapunov candidate can be calculated as follows:
	\begin{equation}
		\dot{L}_c(t) = \eta_c^{-1} \tilde{W}_c \dot{\tilde{W}}_c^T.
	\end{equation}
	From \eqref{eq:critic_weight_update_rule} it can be concluded that:
	\begin{equation}
		\label{eq:Cup2two}
		\begin{aligned}
			\dot{\tilde{W}}_c &= -\dot{\hat{W}}_c = \eta_c \hat{\mathcal{E}}_{TD} \Psi_c(\hat{e}_s(t)) \\ &= \eta_c (\mathcal{E}_{TD} - \tilde{W}_c \Psi_c(\hat{e}_s(t))) \Psi_c(\hat{e}_s(t))^T
		\end{aligned}
	\end{equation}
	According to~\eqref{eq:Cup2two}, $\dot{L}_c(t)$ is obtained as follows:
	\begin{equation}
		\label{eq:eq1}
		\begin{aligned}
			\dot{L}_c(t) = \eta_c^{-1} \left( \eta_c \mathcal{E}_{TD}(t) \tilde{W}_c \Psi(\hat{e}_s(t)) - \eta_c (\tilde{W}_c \Psi(\hat{e}_s(t)))^2 \right)
		\end{aligned}
	\end{equation}
	By applying the Cauchy--Schwarz inequality, the following can be concluded:
	\begin{equation}
		\begin{aligned}
			\dot{L}_c(t) \leq & \frac{1}{2\eta_c} \left( \mathcal{E}_{TD}^2(t) +\|\eta_c\tilde{W}_c^T \Psi(\hat{e}_s(t))\|^2 \right) \\ & - (\tilde{W}_c^T \Psi(\hat{e}_s(t)))^2 \\ \leq& -\frac{\left( 2 - \eta_c \right)}{2} \|\tilde{W}_c^T \Psi_c(\hat{e}_s(t))\|^2 +  \frac{1}{2\eta_c} \mathcal{E}_{TD}^2.
		\end{aligned}
	\end{equation}
    Based on~\eqref{eq:td_error} and by considering Assumption~\ref{as:critic_terms}, the term $\mathcal{E}_{TD}$ is bounded.
    Using the dense property of real numbers, there exists a positive constant $\psi \in (0, \Psi_M]$ to satisfy the following inequality:
    \begin{equation}
		\|\tilde{W}_c^T \Psi(\hat{e}_s(t))\|^2 \leq \|\tilde{W}_c \|^2 \psi^2 \leq \frac{\mathcal{E}_M^2}{\eta_c(2 - \eta_c)},
	\end{equation}
    where $\mathcal{E}_M$ is the upper bound for $\mathcal{E}_{TD}$ and $\Psi_M$ is the upper bound for $\Psi(\hat{e}_s(t))$.
	It can be concluded that $\dot{L}_c(t) < 0$ if: 
	\begin{equation}
		0 < \eta_c < 2, \quad \|\tilde{W}_c^T\| \psi \geq \sqrt{\frac{\mathcal{E}_{TD}^2}{\eta_c(2 - \eta_c)}}.
	\end{equation}
    Therefore, $\tilde{W}_c$ converges to the following compact set:
	\begin{equation}
		\Omega_{\tilde{W}_c} = \left\{\tilde{W}_c: \|\tilde{W}_c\| \leq \frac{\mathcal{E}_M}{\psi \sqrt{\eta_c(2 - \eta_c)}} \right\},
	\end{equation}
    where $\eta_c$ is the learning rate of the critic network, as defined  in~\eqref{eq:critic_weight_update_rule}.
\end{proof}

Lastly, the actor stability is proven as follows:
\begin{theorem}
	For system~\eqref{eq:msystem}, if the weight vector of the critic network $\hat{W}_c$ is updated by~\eqref{eq:critic_weight_update_rule}, the identifier weights are updated by~\eqref{eq:identifier_update_1} and~\eqref{eq:identifier_update_2}, and the approximate optimal control policy is derived from~\eqref{eq:actor_update_1} and~\eqref{eq:actor_update_2}, then, under the control policy $\hat{u}_c(t)$, the state tracking error $e_s(t)$ of the closed-loop system is UUB.
\end{theorem}
\begin{proof}
    Let us use a Lyapunov function candidate as follows:
	\begin{equation}
		\begin{aligned}
			L_a = & \int_{t}^{\infty} \left( \hat{e}_s(\tau)^T Q \hat{e}_s(\tau) + \hat{u}_c(\tau)^T R \hat{u}_c(\tau) \right) d\tau \\ 
			& + \frac{1}{2 \eta_a} \text{tr} \left\{ \tilde{W}_a^T \tilde{W}_a \right\}.
		\end{aligned}
	\end{equation}
    Considering Theorems~\ref{theorem:id1} and~\ref{theorem:identifier}, it can be concluded that:
    \begin{equation}
        L_a = \mathbf{V}(\hat{e}_s(t)) + \frac{1}{2 \eta_a} \text{tr} \left\{ \tilde{W}_a^T \tilde{W}_a \right\}.
    \end{equation}
	where $\mathbf{V}(\cdot)$ is the value function as defined in~\eqref{eq:value_function}. Taking the time derivative of $L_a$ yields:
	\begin{equation}
		\begin{aligned}
			\dot{L}_a &= \frac{\partial \mathbf{V}(\hat{e}_s(t))}{\partial \hat{e}_s(t)} \dot{\hat{e}}_s(t) + \frac{1}{\eta_a} \tilde{W}_a \dot{\tilde{W}}_a^T.
		\end{aligned}
	\end{equation}
	Substituting the actor weight update rules into $\dot{L}_a$ gives:
	\begin{equation}
    \label{eq:lyapanov_der}
		\begin{aligned}
			\dot{L}_a = & \frac{\partial \mathbf{V}(\hat{e}_s(t))}{\partial \hat{e}_s(t)} \dot{\hat{e}}_s(t) \\
			& - \frac{1}{\eta_a} \tilde{W}_a \Big( \Big[ \eta_a \Delta t \hat{W}_c \nabla \Psi_c(\hat{e}_s(t+\Delta t)) \hat{W}_i (I - \Pi_i(t)) \hat{V}_i \\ 
			& \quad + R \hat{u}_c(t) \Big] \sigma(V_a \hat{e}_s(t))\Big)^T. 
		\end{aligned}
	\end{equation}	
	The control update rules are derived by minimizing the Hamiltonian~\eqref{eq:hamiltonian}. Hence, the following holds:
	\begin{equation}
		\label{eq:hamiltonian_balance}
		0 = \hat{e}_s(t)^T Q \hat{e}_s(t) + \hat{u}_c(t)^T R \hat{u}_c(t) + \frac{\partial \mathbf{V}(\hat{e}_s(t))}{\partial \hat{e}_s(t)} \dot{\hat{e}}_s(t).
	\end{equation}
	Integrating~\eqref{eq:hamiltonian_balance} into the Lyapunov derivative in~\eqref{eq:lyapanov_der} gives:
	\begin{equation}
    \label{eq:l_deravative_f}
		\begin{aligned}
			\dot{L}_a = & -\hat{e}_s(t)^T Q \hat{e}_s(t) - \hat{u}_c(t)^T R \hat{u}_c(t) \\
			& - \tilde{W}_a \Big( \Big[ \Delta t \hat{W}_c \nabla \Psi_c(\hat{\hat{e}}_s(t+\Delta t)) \hat{W}_i (I - \Pi_i(t)) \hat{V}_i \\
			& + \frac{1}{\eta_a} R \hat{u}_c(t) \Big] \sigma(V_a \hat{e}_s(t)) \Big)^T.
		\end{aligned}
	\end{equation}
	Based on~\eqref{actor_matrix}, it can be concluded that $\|\hat{u}_c \| \leq \|\hat{W}_a\|$. Then, the residual terms in $\dot{L}_a$ is bounded by $C$ which is defined as:
	\begin{equation}
        \label{eq:constant_C}
		\begin{aligned}
			C = & \Delta t \| \nabla \Psi_c(\hat{e}_s(t+\Delta t))\| \|\tilde{W}_a\| \|\hat{W}_c\| \|\hat{W}_i\| \|\hat{V}_i\| \\ + & \frac{1}{\eta_a} \lambda_{\min}(R) \sigma_m^2 \|\hat{W}_a\| \|\tilde{W}_a\|.
		\end{aligned}
	\end{equation}	
    Substituting~\eqref{eq:constant_C} in~\eqref{eq:l_deravative_f} will result in:
	\begin{equation}
		\dot{L}_a \leq -\lambda_{\min}(Q) \|\hat{e}_s(t)\|^2 - \lambda_{\min}(R) \|\hat{u}_c(t)\|^2 + C.
	\end{equation}
    If $\|\hat{e}_s(t)\| \geq \sqrt{\frac{C}{\lambda_{\min}(Q)}}$,
    then it can be concluded that $\dot{L}_a \leq 0$.
	Therefore, the closed-loop system is UUB, with the bounds on error dynamics as follows:
    \begin{equation}
        \|\hat{e}_s(t)\| \leq \sqrt{\frac{C}{\lambda_{\min}(Q)}}
    \end{equation}
\end{proof}

\textbf{Remark 1}: $\Delta t$ and $\eta_a$ are design parameters that can be used to control the upper bound for tracking error $e_s(t)$.
It is worth mentioning that for cases where an upper bound for $\hat{W}_a$ is not known, the upper bound for $C$ can be obtained by choosing a large number as the upper bound for $\hat{W}_a$. Then, by choosing appropriate design parameters ($\Delta t$ and $\eta_a$), the tracking error remains within the desired range.

\section{Results and Discussion}\label{sec4}
The effectiveness of the proposed controller is validated using two benchmark nonlinear systems: an affine nonlinear single-input, multi-output (SIMO) system~\cite{10387713,7867059} and a nonlinear multi-input, multi-output (MIMO) system~\cite{wang2019adaptive}.
In addition, a case study on power system stability~\cite{10016258} is presented to further demonstrate the practical applicability of the controller.



\subsection{Nonlinear SIMO System}\label{sec:scenario_a}

Consider the following nonlinear SIMO system with two state variables and an affine control input:
\begin{equation}
	\dot{x} = f(x) + g(x)u,
\end{equation}
where $f$ and $g$ are system dynamics matrices, defined as~\cite{7867059}:
\begin{align}
f(x) &= \begin{bmatrix}
	-x_1 + x_2 \\
	-0.5x_1 - 0.5x_2(1 - (\cos(2x_1) + 2)^2)
\end{bmatrix}\\[5pt]
g(x) &= \begin{bmatrix}
	0 \\
	\cos(2x_1) + 2
\end{bmatrix}.
\end{align}
The desired state trajectory for the simulations is given by:
\begin{equation}
x_d=\begin{bmatrix}
	sin(t) \\
	cos(t)+sin(t)
\end{bmatrix}.
\end{equation}

The following quadratic function is used as the value function:
\begin{equation}
    \label{eq:ex1-value}
	\mathbf{V}(e(t), \hat{u}_c(t)) = \int_t^{\infty} \frac{1}{2} e_1^2(\tau) + e_2^2(\tau) + 5\times 10^{-4}\hat{u}^{2}_c(\tau) d\tau
\end{equation}
where $e_1$ and $e_2$ are the state tracking errors and $\hat{u}_c$ is the control command.

The critic network is implemented as a single-layer NN consisting of $3$ neurons with the following activation functions:
\begin{equation}
\label{eq:activations}
\Psi_c(e)=
\left [ 
\begin{matrix}
e_1^2 & e_1e_2 & e_2^2 \\
\end{matrix}
\right ]^T.
\end{equation}


The proposed method is simulated for 20 seconds using the parameters in Table~\ref{table:ex1param}. The actor network uses a larger value for the learning rate to ensure fast adaptation and convergence towards the desired trajectory.

\begin{table}[h]
    \centering
    \caption{Simulation parameters for SIMO system.}
    \vspace{-5pt}
    \label{table:ex1param}
    \renewcommand{\arraystretch}{1.2}
    \resizebox{0.9\linewidth}{!}
    {\begin{tabular}{lccc}
        \toprule
        Parameter & Actor & Identifier & Critic\\
        \midrule
        Number of neurons & $8$ & $8$ & $3$ \\
        Activation & Sigmoid & Sigmoid & $\Psi_c$ \\
        Learning rate & $10^2$ & $10^{-3}$ & $10^{-9}$ \\
        Dropout rate & $\gamma_c = 0.7$ & $\gamma_p = 0.8$ & $-$ \\
        \bottomrule
    \end{tabular}}
\end{table}


The system states are represented in Fig.~\ref{fig:states_nonlinear1} and Fig.\ref{fig:states_nonlinear2}. 
It can be concluded that the states follow the desired trajectory without prior knowledge of the system dynamics. 
Figs.~\ref{fig:states_nonlinear3} and \ref{fig:UUB_ex1} show the convergence of the tracking error in 1D and 2D spaces, respectively.
It can be seen that the tracking errors converge to zero and remain within a bounded range. 
The control command remains bounded as shown in Fig.~\ref{fig:states_nonlinear4}. 
Fig.~\ref{fig:states_nonlinear5} illustrates the packet dropout events.
The blue lines represent binary values, which are zero in the absence of a packet dropout and one otherwise. 
\begin{figure}[t]
    \centering
    \begin{subfigure}[b]{1\linewidth}
        \centering
        \includegraphics[width=1\textwidth]{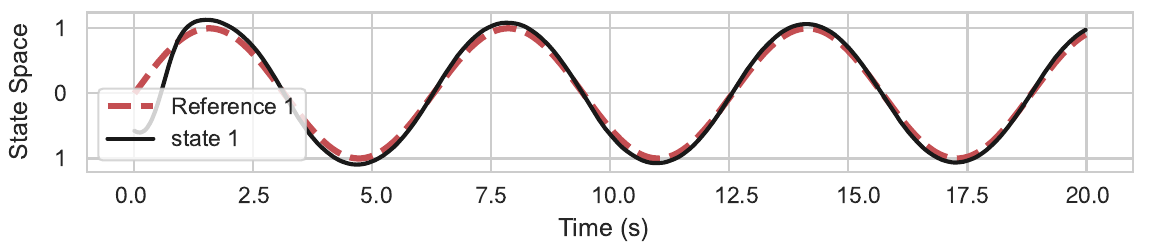}
        \vspace{-18px}
        \caption{First system state.}
        \label{fig:states_nonlinear1}
        \vspace{2px}
    \end{subfigure}
    \hfill
    \begin{subfigure}[b]{1\linewidth}
        \centering
        \includegraphics[width=1\textwidth]{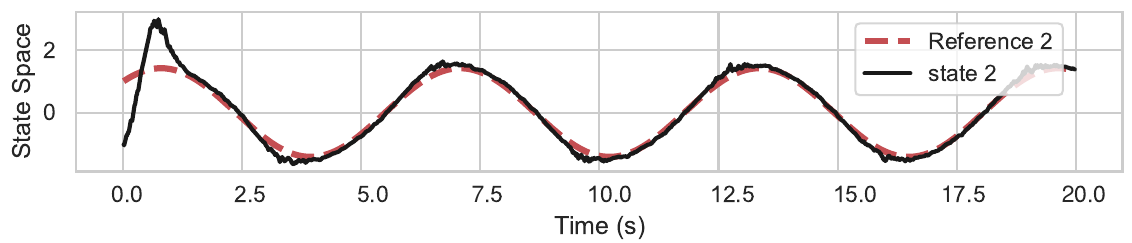}
        \vspace{-18px}
        \caption{Second system state.}
        \label{fig:states_nonlinear2}
        \vspace{2px}
    \end{subfigure}
    \hfill
    \begin{subfigure}[b]{1\linewidth}
        \centering
        \includegraphics[width=1\textwidth]{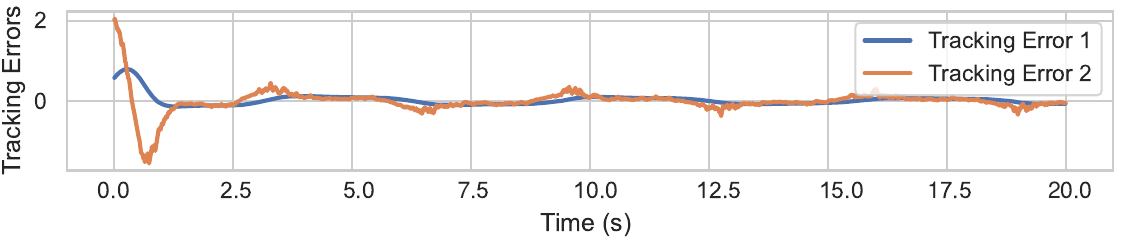}
        \vspace{-18px}
        \caption{State tracking errors.}
        \label{fig:states_nonlinear3}
        \vspace{2px}
    \end{subfigure}
    \hfill
    \begin{subfigure}[b]{1\linewidth}
        \centering
        \includegraphics[width=1\textwidth]{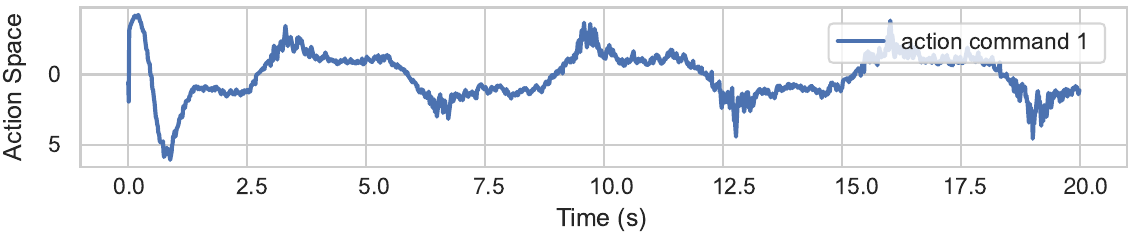}
        \vspace{-18px}
        \caption{Control command.}
        \label{fig:states_nonlinear4}
        \vspace{2px}
    \end{subfigure}
    \hfill
    \begin{subfigure}[b]{1\linewidth}
        \centering
        \includegraphics[width=1\textwidth]{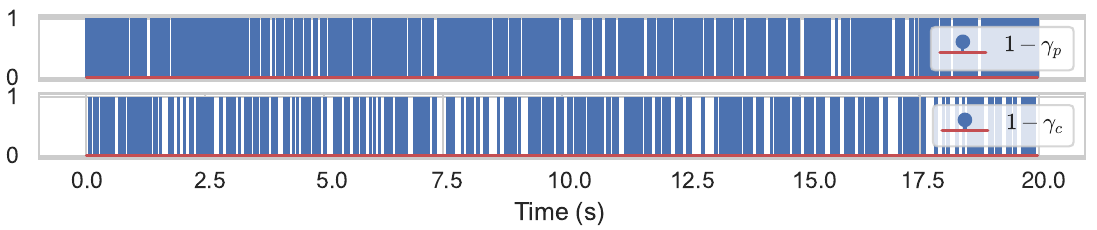}
        \vspace{-18px}
        \caption{Packet dropout events.}
        \label{fig:states_nonlinear5}
    \end{subfigure}
    \vspace{-15px}
    \caption{Tracking control results for SIMO system.}
    \label{fig:states_nonlinear}
    \vspace{-5px}
\end{figure}

\begin{figure}[h]
    \centering
    \includegraphics[width=.95\linewidth]{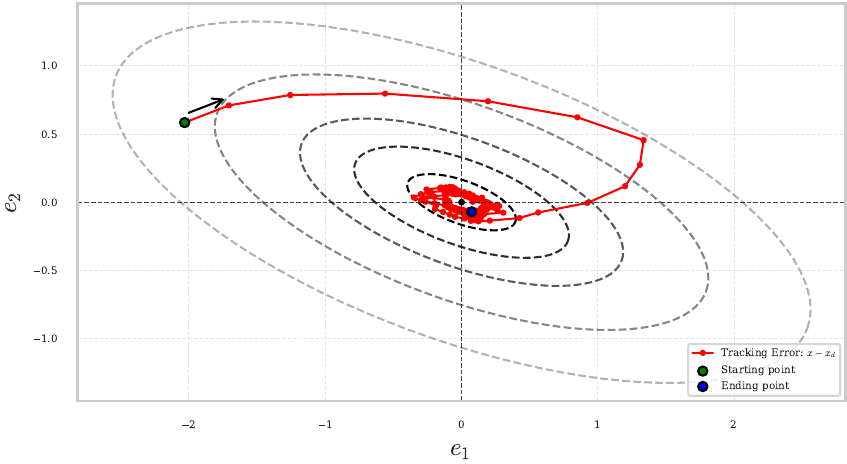}
    \caption{SIMO system tracking error in 2D space.}
    \label{fig:UUB_ex1}
\end{figure}

Fig.~\ref{fig:Results_nonlinear} shows the training error for all three networks, i.e., the actor, identifier, and critic.
The error values converge quickly to a bounded range around zero in the presence of packet dropouts. As illustrated in Fig.~\ref{fig:Results_nonlinear4}, the gradients of the actor parameters converge and become stable in the control process.
\begin{figure}[t]
    \centering
    \begin{subfigure}[b]{0.95\linewidth}
        \centering
        \includegraphics[width=1\textwidth]{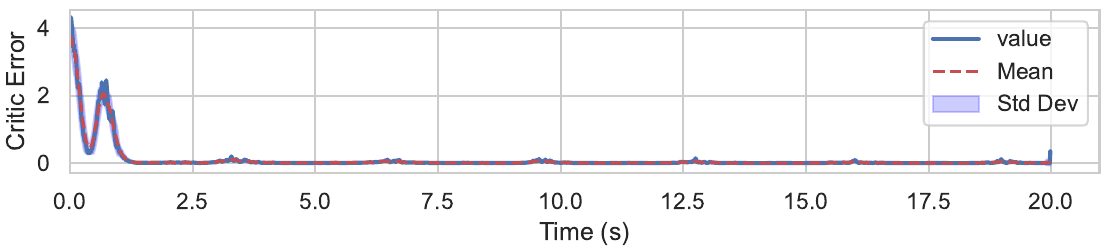}
        \vspace{-18px}
        \caption{Critic error.}
        \label{fig:Results_nonlinear1}
        \vspace{2px}
    \end{subfigure}
    \hfill
    \begin{subfigure}[b]{0.95\linewidth}
        \centering
        \includegraphics[width=1\textwidth]{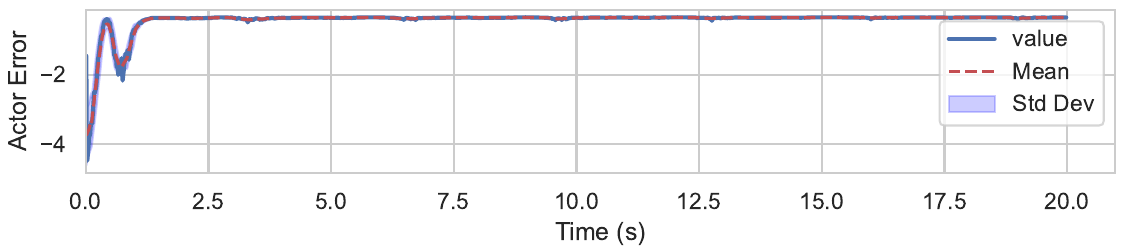}
        \vspace{-18px}
        \caption{Actor error.}
        \label{fig:Results_nonlinear2}
        \vspace{2px}
    \end{subfigure}
    \hfill
    \begin{subfigure}[b]{0.95\linewidth}
        \centering
        \includegraphics[width=1\textwidth]{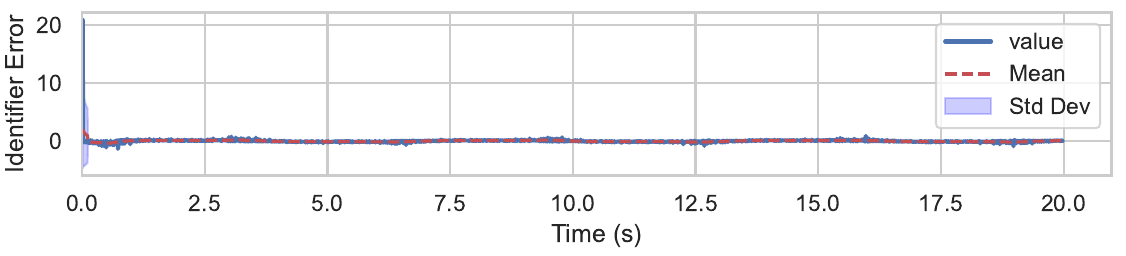}
        \vspace{-18px}
        \caption{Identifier error.}
        \label{fig:Results_nonlinear3}
        \vspace{2px}
    \end{subfigure}
    \hfill
    \begin{subfigure}[b]{0.97\linewidth}
        \centering
        \includegraphics[width=1\textwidth]{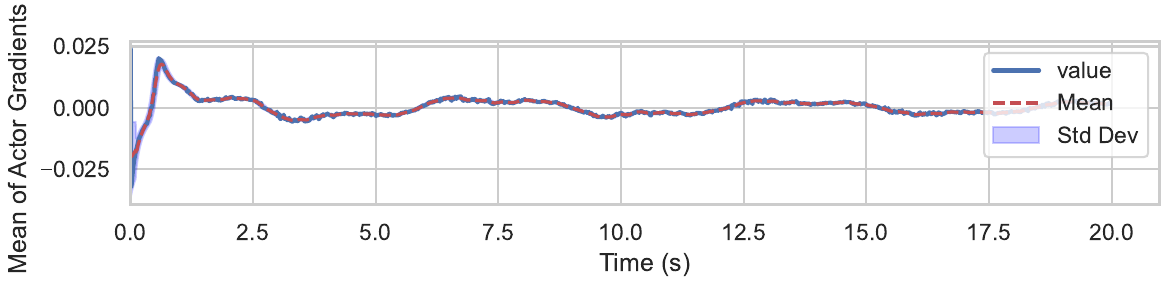}
        \vspace{-18px}
        \caption{Mean of actor gradients.}
        \label{fig:Results_nonlinear4}
    \end{subfigure}
    \caption{Convergence error for AIC controller components.}
    \label{fig:Results_nonlinear}
    \vspace{-5px}
\end{figure}
The value space estimated by the critic is illustrated in Fig.~\ref{fig:Value_space_nonlinear}. This value function has a quadratic form as a result of choosing proper activation functions. 
\begin{figure}[h]
    \centering
    \includegraphics[width=.7\linewidth]{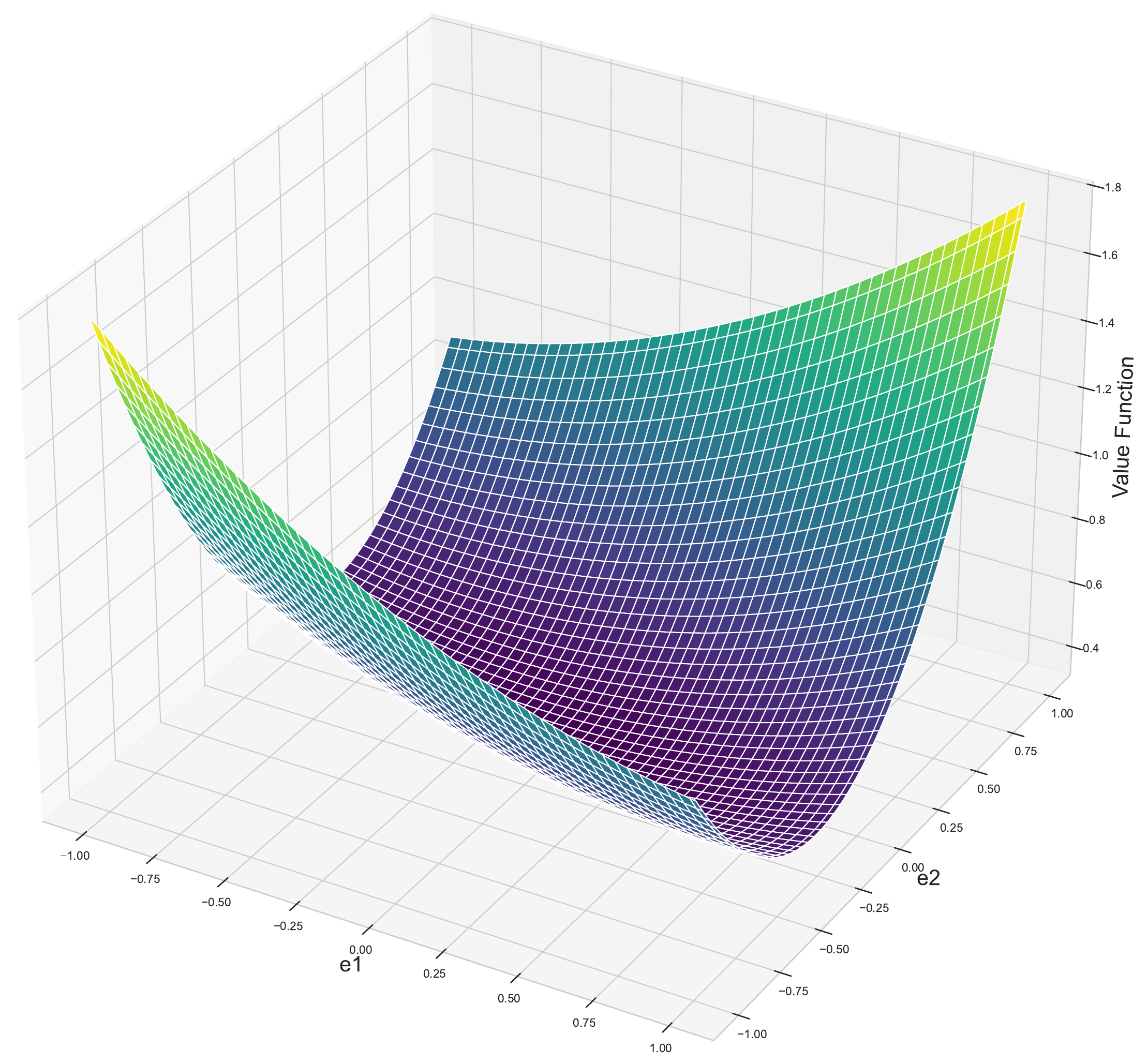}
    \caption{Estimated value space by critic for SIMO system.}
    \label{fig:Value_space_nonlinear}
\end{figure}

To show the identifier performance, the direction of the estimated function $\hat{g}(\cdot)$ is compared with the direction of $g(\cdot)$, which is defined in~\eqref{eq:sys_dynamic}. As shown in Fig.~\ref{fig:identifier_stability_ex1}, $\hat{g}(\cdot)$ and $g(\cdot)$ are in the same direction, which confirms the effectiveness of the identifier.
\begin{figure}[h]
    \centering
    \includegraphics[width=.95\linewidth]{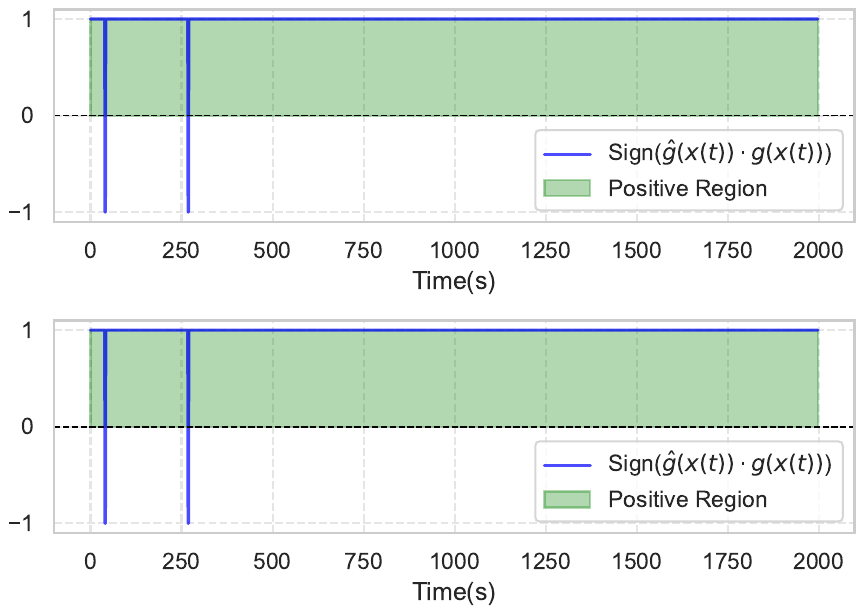}
    \caption{Alignment of identifier gradient direction with the SIMO system dynamics direction.}
    \label{fig:identifier_stability_ex1}
\end{figure}

\subsection{Nonlinear MIMO System}\label{sec:scenario_b}
In the second scenario, the proposed control algorithm is applied to a nonlinear MIMO system with two states and two control inputs. The matrices of the system dynamics are defined as follows~\cite{wang2019adaptive}:
\begin{equation}
f(x) = \begin{bmatrix} 
        x_2 - x_1 \\ 
        0.5(x_1 x_2 - x_1) 
      \end{bmatrix} 
, \quad
g(x) = \begin{bmatrix} 
        0 & 3 + x_2 \\ 
        1 + x_1 & 0 
      \end{bmatrix} 
\end{equation}
The desired state trajectory is considered based on~\cite{wang2019adaptive}.
The value function for the state tracking problem is considered as follows:
\begin{equation}
    \label{eq:ex1-value}
	\mathbf{V}(e(t), \hat{u}_c(t)) = \int_t^{\infty} \frac{1}{2} e_1^2(\tau) + e_2^2(\tau) + 5\times 10^{-3}\hat{u}^{2}_c(\tau) d\tau,
\end{equation}
The activation functions of the critic are defined in~\eqref{eq:ex1-value}.
The control method is simulated for 30 seconds with parameters represented in Table~\ref{table:ex2param}.
\begin{table}[H]
    \centering
    \caption{Simulation parameters for MIMO system.}
    \vspace{-5pt}
    \label{table:ex2param}
    \renewcommand{\arraystretch}{1.2}
    \resizebox{0.9\linewidth}{!}
    {\begin{tabular}{lccc}
        \toprule
        Parameter & Actor & Identifier & Critic\\
        \midrule
        Number of neurons & $64$ & $2$ & $3$ \\
        Activation & Sigmoid & Sigmoid & $\Psi_c$\\
        Learning rate & $10$ & $10^-3$ & $10^-3$ \\
        Dropout rate & $\gamma_c = 0.7$ & $\gamma_p = 0.8$ & $-$ \\
        \bottomrule
    \end{tabular}}
\end{table}


The simulation results for the tracking process are represented in Fig.~\ref{fig:states_MIMO}. It is shown that system states track the desired trajectory. 
It can be concluded that the tracking error converges to a bounded range.

\begin{figure}[t]
    \centering
    \begin{subfigure}[b]{1\linewidth}
        \centering
        \includegraphics[width=1\textwidth]{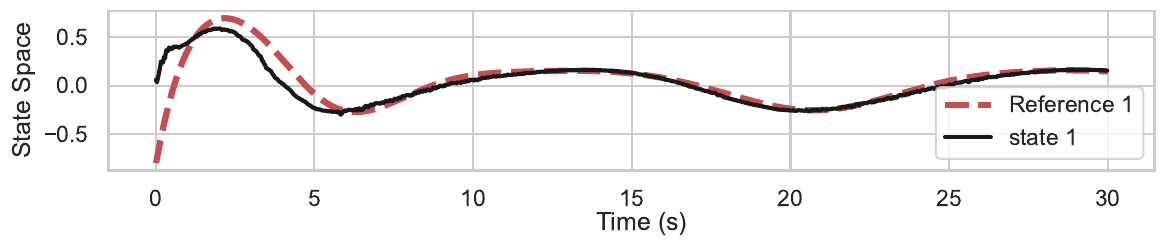}
        \vspace{-18px}
        \caption{First system state.}
        \label{fig:states_MIMO1}
        \vspace{2px}
    \end{subfigure}
    \hfill
    \begin{subfigure}[b]{1\linewidth}
        \centering
        \includegraphics[width=1\textwidth]{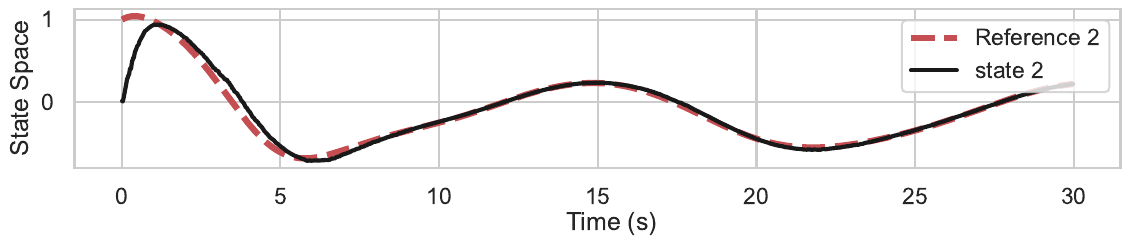}
        \vspace{-18px}
        \caption{Second system state.}
        \label{fig:states_MIMO2}
        \vspace{2px}
    \end{subfigure}
    \hfill
    \begin{subfigure}[b]{1\linewidth}
        \centering
        \includegraphics[width=1\textwidth]{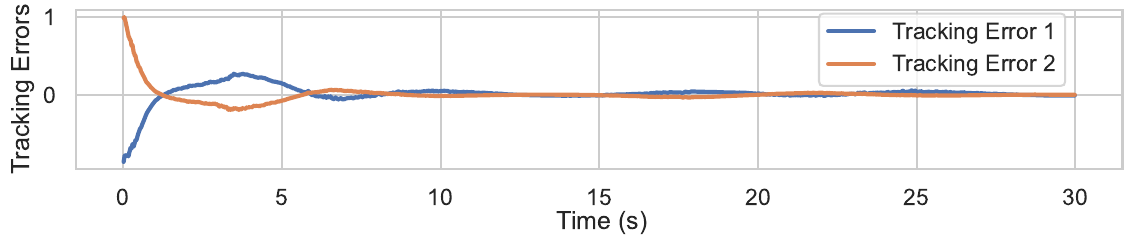}
        \vspace{-18px}
        \caption{State tracking errors.}
        \label{fig:states_MIMO3}
        \vspace{2px}
    \end{subfigure}
    \hfill
    \begin{subfigure}[b]{1\linewidth}
        \centering
        \includegraphics[width=1\textwidth]{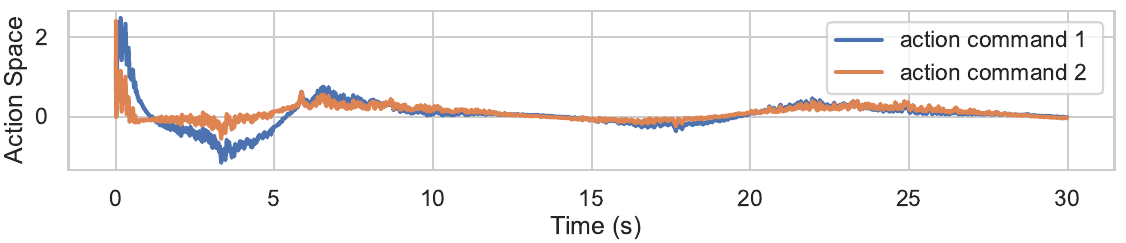}
        \vspace{-18px}
        \caption{Control commands.}
        \label{fig:states_MIMO4}
        \vspace{2px}
    \end{subfigure}
    \hfill
    \begin{subfigure}[b]{1\linewidth}
        \centering
        \includegraphics[width=1\textwidth]{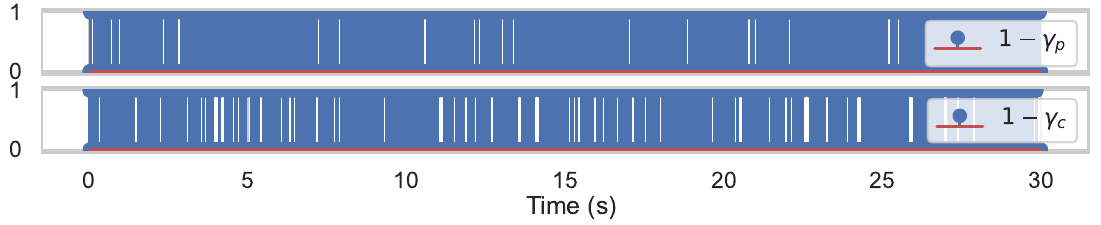}
        \vspace{-18px}
        \caption{Packet dropout events.}
        \label{fig:states_MIMO5}
    \end{subfigure}
    \caption{Tracking control results for MIMO system.}
    \label{fig:states_MIMO}
\end{figure}


The error values for AIC controller networks are shown in Fig.~\ref{fig:Results_MIMO}. As expected, the networks are stable and the error values converge to their respective bounded range.

\begin{figure}[h]
    \centering
    \begin{subfigure}[b]{1\linewidth}
        \centering
        \includegraphics[width=1\textwidth]{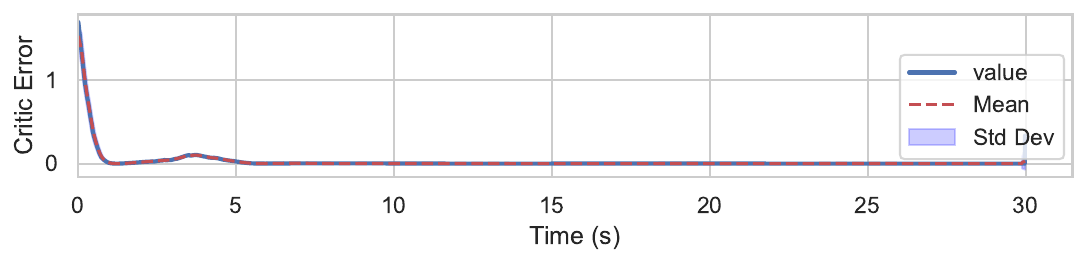}
        \vspace{-18px}
        \caption{Critic error.}
        \label{fig:Results_MIMO1}
        \vspace{2px}
    \end{subfigure}
    \hfill
    \begin{subfigure}[b]{1\linewidth}
        \centering
        \includegraphics[width=1\textwidth]{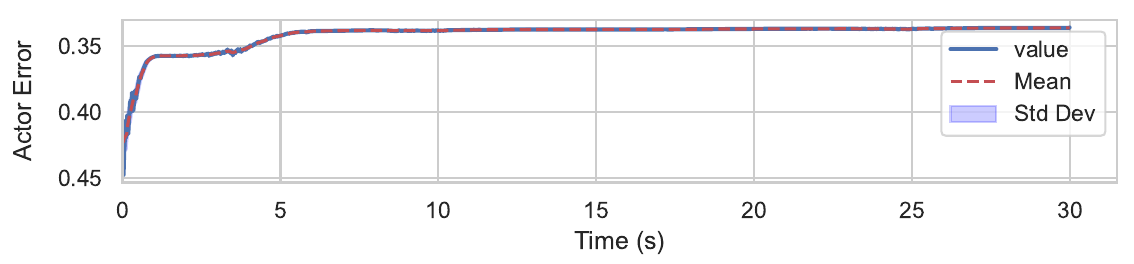}
        \vspace{-18px}
        \caption{Actor error.}
        \label{fig:Results_MIMO2}
        \vspace{2px}
    \end{subfigure}
    \hfill
    \begin{subfigure}[b]{1\linewidth}
        \centering
        \includegraphics[width=1\textwidth]{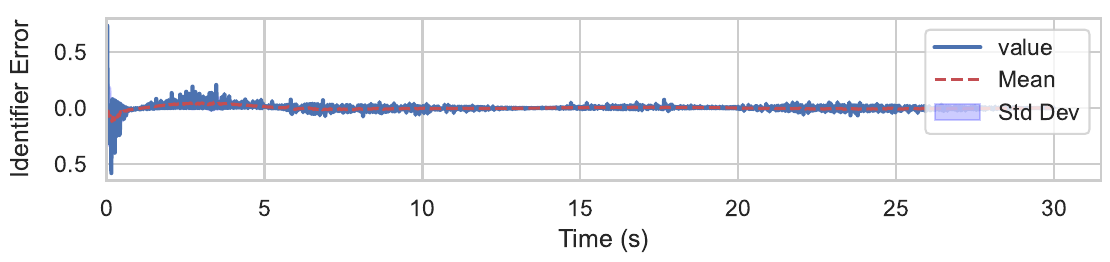}
        \vspace{-18px}
        \caption{Identifier error.}
        \label{fig:Results_MIMO3}
        \vspace{2px}
    \end{subfigure}
    \hfill
    \begin{subfigure}[b]{1\linewidth}
        \centering
        \includegraphics[width=1\textwidth]{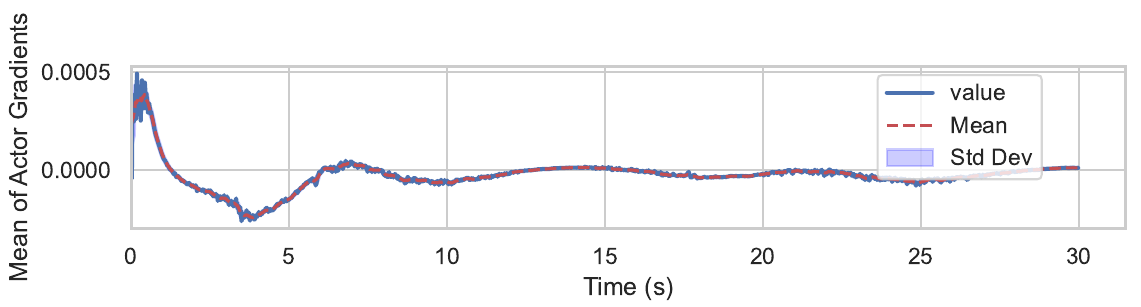}
        \vspace{-18px}
        \caption{Mean of actor gradients.}
        \label{fig:Results_MIMO4}
    \end{subfigure}
    \caption{Convergence error for AIC controller components.}
    \label{fig:Results_MIMO}
    \vspace{-10px}
\end{figure}

The estimated value space by the critic network is shown in Fig.~\ref{fig:Value_space_MIMO}. As shown in Fig.~\ref{fig:UUB_ex2}, the actor uses this value estimate to optimize the control policy.
The identifier performance is represented in Fig.\ref{fig:identifier_stability_ex2}. It can be seen that the gradient of the identifier is in the correct direction for all states.

\begin{figure}[h]
    \centering
    \begin{subfigure}[b]{0.55\linewidth}
        \centering
        \includegraphics[width=1\textwidth]{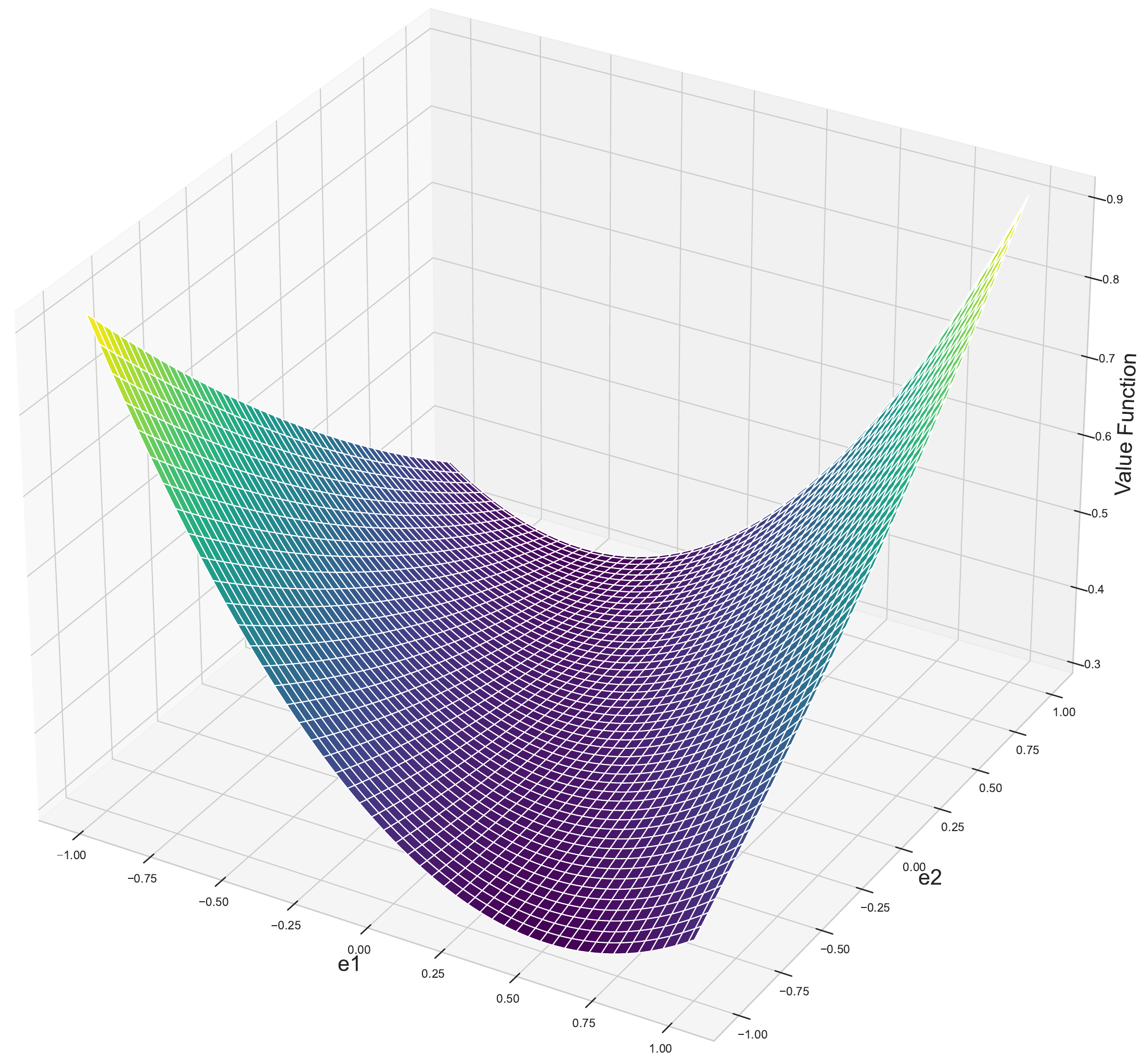}
        \vspace{-15px}
        \caption{Estimated value space.}
        \label{fig:Value_space_MIMO}
        \vspace{2px}
    \end{subfigure}
    \hfill
    \begin{subfigure}[b]{0.43\linewidth}
        \centering
        \includegraphics[width=1\textwidth]{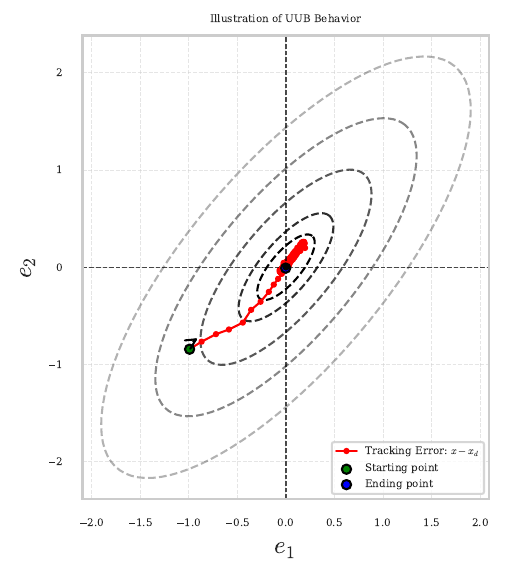}
        \vspace{-15px}
        \caption{System tracking error.}
        \label{fig:UUB_ex2}
        \vspace{2px}
    \end{subfigure}
    \caption{Value estimate and its respective policy behavior.}
    \label{fig:Results_MIMO}
    \vspace{-10px}
\end{figure}


\begin{figure}[h]
    \centering
    \includegraphics[width=1\linewidth]{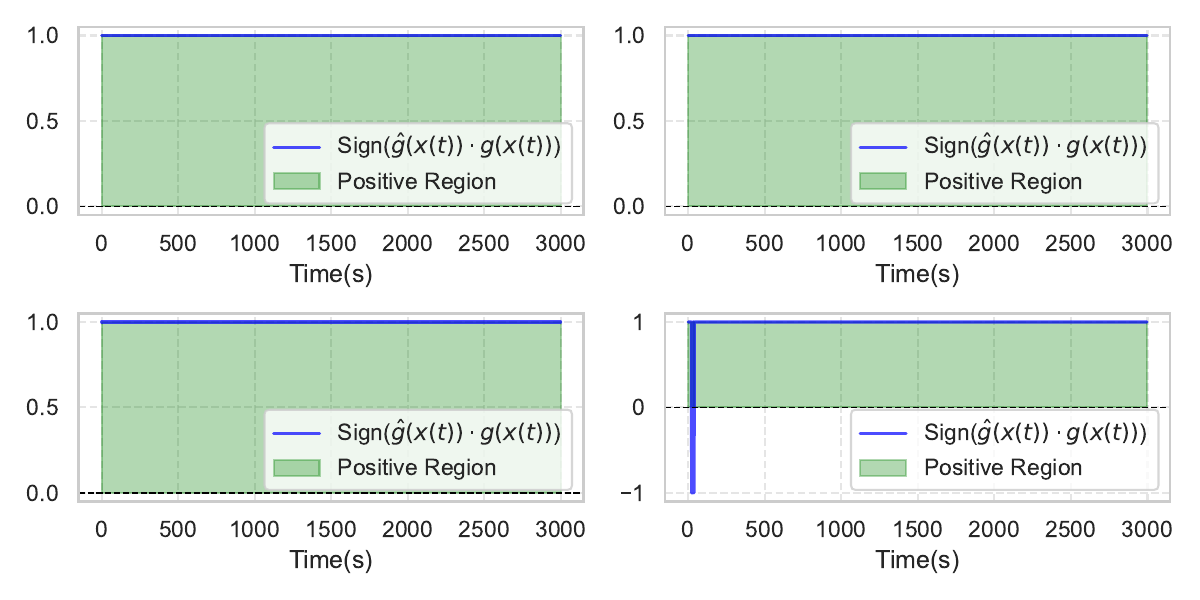}
    \caption{Alignment of identifier gradient direction with the MIMO system dynamics direction.}
    \label{fig:identifier_stability_ex2}
\end{figure}

\subsection{Quantitative analysis}

The performance of the AIC controller for the systems described in~\ref{sec:scenario_a} and~\ref{sec:scenario_b} is evaluated for five dropout scenarios. These scenarios are considered as: no dropouts, dropouts on the control-to-plant communication side, dropouts on the plant-to-control communication side, dropouts on both sides with low and equal occurrence probabilities, and dropouts on both sides with high but different occurrence probabilities. 


We measure two performance criteria for the evaluation: Normalized Root Mean Square Error (NRMSE) and the Pearson Correlation Coefficient (PCC).
The NRMSE is defined as:
\begin{equation}
    \label{eq:NRMSE}
    \text{NRMSE} = 
    \frac{
    \sqrt{
    \frac{1}{2N} \sum_{i=1}^{N} \sum_{j=1}^{2} (y_{i,j} - \hat{y}_{i,j})^2
    }
    }{
    \frac{1}{2} \sum_{j=1}^{2} \left( \underset{\scalebox{.5}{$1 \leq i \leq N$}}{\max} y_{i,j} - \underset{\scalebox{.5}{$1 \leq i \leq N$}}{\min} y_{i,j} \right)
    }
\end{equation}

where $y_{i,j}$ and $\hat{y}_{i,j}$ represent actual and desired state trajectories, respectively. $N$ is the number of steps in the simulation. The summation is applied to 2 states of the system.

The PCC criteria is calculated as:
\begin{equation}
    \text{PCC} = \frac{\sum_{i=1}^{N} (y_i - \bar{y})(\hat{y}_i - \bar{\hat{y}})}{\sqrt{\sum_{i=1}^{N} (y_i - \bar{y})^2} \sqrt{\sum_{i=1}^{N} (\hat{y}_i - \bar{\hat{y}})^2}}
\end{equation}
where $\bar{y}$ and $\bar{\hat{y}}$ are the mean values for actual and desired state trajectories.

\begin{table*}[b]
	\centering
	\caption{Tracking error results for SIMO and MIMO systems.}
	\renewcommand{\arraystretch}{1.3}
	\begin{tabular}{cccccccccccccc}
		\toprule
		\multirow{3}{*}{\textbf{System}} & \multirow{3}{*}{\textbf{}} & \multirow{3}{*}{\textbf{}} & \multicolumn{2}{c}{$\bar{\gamma}_s = 1 \quad \bar{\gamma}_c = 1$} & \multicolumn{2}{c}{$\bar{\gamma}_s = 1 \quad \bar{\gamma}_c = 0.8$} & \multicolumn{2}{c}{$\bar{\gamma}_s = 0.8 \quad \bar{\gamma}_c = 1$} & \multicolumn{2}{c}{$\bar{\gamma}_s = 0.9 \quad \bar{\gamma}_c = 0.9$} & \multicolumn{2}{c}{$\bar{\gamma}_s = 0.8 \quad \bar{\gamma}_c = 0.7$} \\
		\cmidrule(lr){4-5} \cmidrule(lr){6-7} \cmidrule(lr){8-9} \cmidrule(lr){10-11} \cmidrule(lr){12-13}
		& & & \textbf{NRMSE} & \textbf{PCC} & \textbf{NRMSE} & \textbf{PCC} & \textbf{NRMSE} & \textbf{PCC} & \textbf{NRMSE} & \textbf{PCC} & \textbf{NRMSE} & \textbf{PCC} \\
		\midrule
		
		\multirow{7}{*}{\rotatebox[origin=c]{90}{SIMO}} & \multirow{3}{*}{\rotatebox[origin=c]{90}{OTE}} & MAX & 0.2442 & 0.9868 & 0.2152 & 0.9917 & 0.1233 & 0.9945 & 0.1715 & 0.9927 & 0.1931 & 0.9916 \\
		&  & MEAN & 0.1611 & 0.9725 & 0.1458 & 0.9756 & 0.0998 & 0.9885 & 0.1239 & 0.9815 & 0.1248 & 0.9799 \\
		& & STD & 0.0453 & 0.0127 & 0.0438 & 0.0108 & 0.0191 & 0.0036 & 0.0302 & 0.0080 & 0.0456 & 0.0124 \\
		 
		\cmidrule(lr){2-13}
		& \multirow{3}{*}{\rotatebox[origin=c]{90}{PCTE}}  & MAX & 0.0737 & 0.9999 & 0.0684 & 0.9997 & 0.0455 & 0.9999 & 0.0287 & 0.9998 & 0.0340 & 0.9997 \\
		& & MEAN & 0.0389 & 0.9988 & 0.0297 & 0.9993 & 0.0207 & 0.9998 & 0.0237 & 0.9996 & 0.0281 & 0.9995 \\
		& & STD & 0.0200 & 0.0009 & 0.0178 & 0.0006 & 0.0114 & 0.0001 & 0.0042 & 0.0001 & 0.0047 & 0.0001 \\
		\cmidrule(lr){2-13}
		& \multicolumn{2}{l}{Settle Time} &  \multicolumn{2}{c}{2.224 sec} &  \multicolumn{2}{c}{2.124 sec}&  \multicolumn{2}{c}{1.792 sec}  &  \multicolumn{2}{c}{1.937 sec}  & \multicolumn{2}{c}{1.871 sec} \\
		\midrule
		
		\multirow{7}{*}{\rotatebox[origin=c]{90}{MIMO}} & \multirow{3}{*}{\rotatebox[origin=c]{90}{OTE}} & MAX & 0.1404 & 0.9759 & 0.2447 & 0.9655 & 0.1895 & 0.9556 & 0.2651 & 0.9506 & 0.2837 & 0.9253 \\
		& & MEAN & 0.1125 & 0.9625 & 0.1745 & 0.9153 & 0.1729 & 0.9404 & 0.1935 & 0.9231 & 0.2301 & 0.8979 \\
		& & STD & 0.0183 & 0.0092 & 0.0629 & 0.0390 & 0.0162 & 0.0085 & 0.0383 & 0.0280 & 0.0300 & 0.0237 \\ 
		\cmidrule(lr){2-13}
		& \multirow{3}{*}{\rotatebox[origin=c]{90}{PCTE}}  & MAX &  0.0599 & 0.9975 & 0.0695 & 0.9984 & 0.0823 & 0.9963 & 0.0766 & 0.9993 & 0.0774 & 0.9952 \\
		& & MEAN & 0.0449 & 0.9965 & 0.0551 & 0.9943 & 0.0666 & 0.9937 & 0.0512 & 0.9946 & 0.0670 & 0.9926 \\ 
		& & STD & 0.0075 & 0.0013 & 0.0135 & 0.0021 & 0.0144 & 0.0022 & 0.0198 & 0.0027 & 0.0101 & 0.0014 \\
		\cmidrule(lr){2-13}
		& \multicolumn{2}{l}{Settle Time} &  \multicolumn{2}{c}{5.796 sec} &  \multicolumn{2}{c}{6.231 sec}&  \multicolumn{2}{c}{6.424 sec}  &  \multicolumn{2}{c}{5.037 sec}  & \multicolumn{2}{c}{5.278 sec} \\
		\bottomrule
	\end{tabular}
    \label{tab:eval1}
\end{table*}

The NRMSE and PCC are used to assess the accuracy and correlation of the system output, respectively. Each experiment is repeated 5 times for statistical robustness, with the maximum, mean, and standard deviation (STD) values reported from these repetitions. 
The results are presented for two specific scenarios: overall tracking error (OTE) and post-convergence tracking error (PCTE). The OTE represents the system performance for the entire simulation time, while the PCTE shows the controller performance in maintaining accurate tracking after convergence. 

The simulation results in Table~\ref{tab:eval1} show that the NRMSE is low for all dropout scenarios. Moreover, the PCC is close to $1$ in all dropout scenarios. It can be concluded that, despite dropouts, tracking remains within a small bounded range.
The small STD in the simulations demonstrates the reliable performance of the AIC controller. 

\subsection{Case study: dynamic stability in power systems}
This case study analyzes the dynamic behavior of a power system with high IBR.
In conventional power systems, the rotational inertia of synchronous generators helps stabilize frequency during sudden power imbalances. However, inverter-based systems lack this inherent inertia; therefore, a virtual synchronous machine (VSM) controller is employed to emulate synchronous generator dynamics and maintain transient stability of the grid.
In a general representation of a nonlinear system~\eqref{eq:sys_dynamic}, the dynamic matrices $f$ and $g$ of the VSM are given by~\cite{10016258}:
\begin{align}
&f(x) = 
\begin{bmatrix}
\omega - \omega_{\text{nom}} \\
\frac{1}{H} \left( -D_p(\omega - \omega_{\text{nom}}) - \frac{EV_c}{X_\text{eq}} \sin(\delta) \right)
\end{bmatrix}\\[5pt]
&g(x) = 
\begin{bmatrix}
0 \\
\frac{1}{H}
\end{bmatrix}
\end{align}
Here, $\omega_\text{nom}$ is the nominal grid frequency, $D_p$ is the damping factor, and $H$ is the virtual inertia coefficient which is set to a small value, making the system more sensitive to power imbalances and prone to instability. The maximum dispatchable power is given by $P_\text{max}=\frac{EV_c}{X_\text{eq}}$, where $E$ is the inverter output voltage, $V_c$ is the grid voltage, and $X_\text{eq}$ is the equivalent reactance.
The system frequency is initialized at $50$~Hz, and the controller aims to maintain stability at $\omega_d = \omega_\text{nom} = 50$~Hz after a power imbalance.

\begin{figure}[h]
    \centering
    \begin{subfigure}[b]{1\linewidth}
        \centering
        \includegraphics[width=1\textwidth]{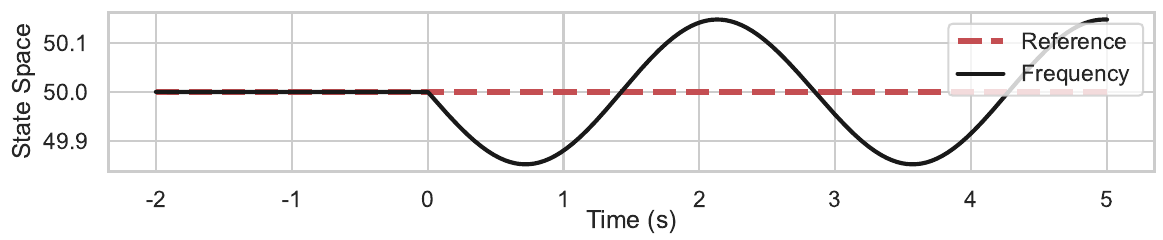}
        \vspace{-18px}
        \caption{No control.}
        \label{fig:vsm_noc}
        \vspace{2px}
    \end{subfigure}
    \hfill
    \begin{subfigure}[b]{1\linewidth}
        \centering
        \includegraphics[width=1\textwidth]{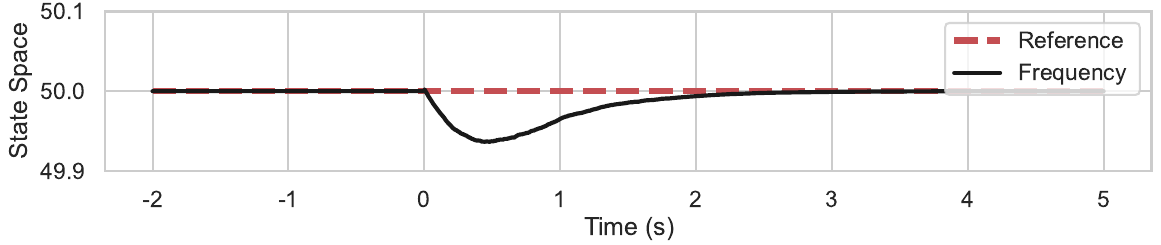}
        \vspace{-18px}
        \caption{AIC controller.}
        \label{fig:vsm_aic}
        \vspace{2px}
    \end{subfigure}
    \hfill
    \begin{subfigure}[b]{1\linewidth}
        \centering
        \includegraphics[width=1\textwidth]{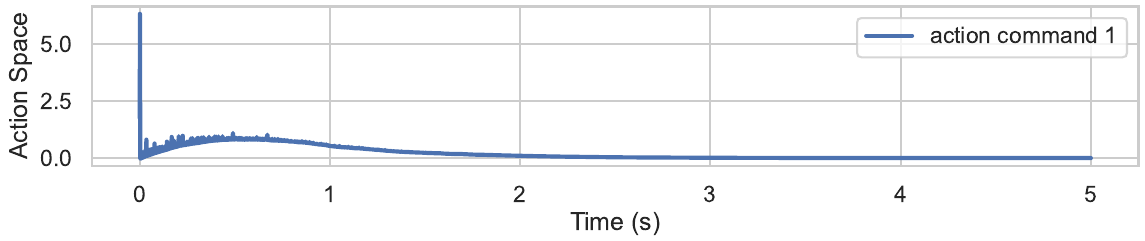}
        \vspace{-18px}
        \caption{Control commands.}
        \label{fig:vsm_action}
        \vspace{2px}
    \end{subfigure}
    \hfill
    \begin{subfigure}[b]{1\linewidth}
        \centering
        \includegraphics[width=1\textwidth]{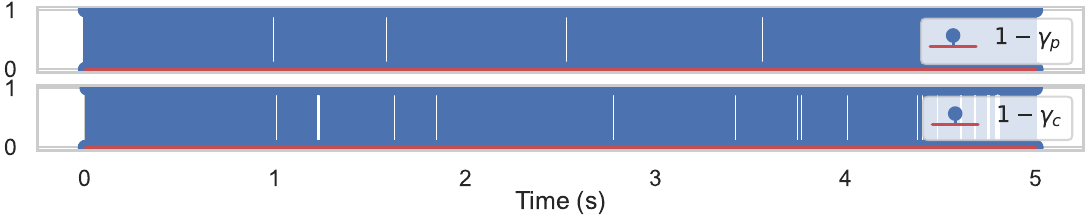}
        \vspace{-18px}
        \caption{Packet dropout events.}
        \label{fig:vsm_dropout}
    \end{subfigure}
    \caption{VSM frequency control.}
    \label{fig:vsm_results}
\end{figure}

The results of VSM frequency control are presented in Fig.~\ref{fig:vsm_results}. A power imbalance occurs at~$t=0$, which causes frequency deviations. Due to a lack of rotational inertia in VSM, the frequency deviation increases and becomes unstable~(\ref{fig:vsm_noc}). When the proposed AIC controller is applied, the frequency is restored to $50$~Hz and remains stable~(\ref{fig:vsm_aic}). This is achieved without prior knowledge of the system model and in the presence of packet dropouts~(\ref{fig:vsm_dropout}), which confirms the applicability of the proposed AIC controller.
\section{Conclusion}\label{sec5}
This study investigates the nonlinear tracking control problem in the presence of packet dropouts by proposing a novel model-free adaptive AIC controller. The presented strategy uses a triple-network structure, consisting of actor, identifier, and critic. Using an NN-based identifier, the controller estimates unknown system dynamics in the presence of stochastic packet dropouts and facilitates gradient transfer from the critic to the actor.
The optimality of the controller is studied using the HJB equation, and its stability analysis is provided based on Lyapunov’s stability theorem.
The controller performance is evaluated over three case studies, SIMO, MIMO, and a power system case study in the presence of packet dropouts. 
The results demonstrate that the proposed AIC controller effectively controls unknown nonlinear systems with bounded tracking error. Future research will address other networked control system challenges, such as communication delay.

\printbibliography

 




\end{document}